\documentclass[10pt]{amsart}
\usepackage{a4wide}
\usepackage{amsfonts}
\usepackage{psfrag}
\usepackage{color}
\usepackage[all]{xypic}
\usepackage{hyperref}
\usepackage{amsmath}
\usepackage{amssymb}
\usepackage{amsthm}
\usepackage{graphicx}
\usepackage{mathtools}
\usepackage{dsfont}
\newtheorem{theorem}{Theorem}[section]
\newtheorem{corollary}[theorem]{Corollary}
\newtheorem{lemma}[theorem]{Lemma}
\newtheorem{proposition}[theorem]{Proposition}

\theoremstyle{definition}
\newtheorem{definition}[theorem]{Definition}
\newtheorem{remark}[theorem]{Remark}

\numberwithin{equation}{section}
\numberwithin{equation}{section}
\numberwithin{equation}{section}

\newcommand{\V}{\mathcal{V}}
\newcommand{\Hh}{\mathcal{H}}
\newcommand{\C}{\mathcal{C}}
\newcommand{\PP}{\mathbb{P}}

\newcommand{\E}{\mathrm{e}}
\newcommand{\D}{\mathrm{d}}

\newcommand{\R}{\mathbb{R}}
\newcommand{\Ff}{\mathcal{F}}
\newcommand{\Ll}{\mathcal{L}}

\newcommand{\eps}{\epsilon}

\newcommand{\cE}{\mathcal{E}}
\newcommand{\DcE}{D(\cE)}

\begin{document}
\title[Dirichlet Forms and Finite Element Methods for the SABR Model]{Dirichlet Forms and Finite Element Methods for the \\ SABR Model}

\author{Blanka Horvath}
\address{Department of Mathematics, Imperial College London}
\email{b.horvath@imperial.ac.uk}

\author{Oleg Reichmann}
\address{European Investment Bank}
\email{o.reichmann@eib.org}
\date{\today}
\thanks{BH is grateful for financial support from the Swiss National Science Foundation (SNSF Grant 165248). The authors are also indebted to the anonymous referees for their helpful comments and to Robbin Tops for several invaluable discussions. The views expressed in this article are those of the authors and do not necessarily represent those of the European Investment Bank.}

\keywords{SABR model, Finite Element Methods, Dirichlet Forms}
\subjclass[2010]{35K15, 65M12, 65M60, 91G30}

\maketitle
\begin{abstract}
We propose a deterministic numerical method for pricing vanilla options under the SABR stochastic volatility model, based on a finite element discretization of the Kolmogorov pricing equations via non-symmetric Dirichlet forms. Our pricing method is valid under mild assumptions on parameter configurations of the process both in moderate interest rate environments and in near-zero interest rate regimes such as the currently prevalent ones.
The parabolic Kolmogorov pricing equations for the SABR model are degenerate at the origin, yielding non-standard partial differential equations, for which conventional pricing methods ---designed for non-degenerate parabolic equations--- potentially break down. 
We derive here the appropriate analytic setup to handle the degeneracy of the model at the origin. That is, we construct 
 an evolution triple of suitably chosen Sobolev spaces with singular weights, consisting of the domain of the SABR-Dirichlet form, its dual space, and the pivotal Hilbert space. In particular, we show well-posedness of the variational formulation of the SABR-pricing equations for vanilla and barrier options on this triple. Furthermore, we present a finite element discretization scheme 
based on a (weighted) multiresolution wavelet approximation in space and a $\theta$-scheme in time and provide an error analysis for this discretization.\\
\end{abstract}


\maketitle

\section{Introduction}
The stochastic alpha beta rho (SABR) model introduced by Hagan et. al 
in \cite{ManagingSmileRisk,HLW} is today industry standard in interest rate markets.
The model with parameters $\nu>0$, $\beta \in [0,1]$, and $\rho\in [-1,1]$, is defined by the pair of coupled stochastic differential equations
\begin{equation}\label{eq:SABRSDE}
\begin{array}{rll}
\D X_t & = Y_t X_t^{\beta}\D W_t, \qquad & X_0 = x_0> 0,\\
\D Y_t & = \nu Y_t \D Z_t, \qquad & Y_0 = y_0>0,\\
\D \langle Z,W\rangle_t & = \rho \D t,\qquad & 0\leq t\leq T<\infty,
\end{array}
\end{equation}
where $W$ and $Z$ are  $\rho$-correlated Brownian motions on a filtered probability space
$(\Omega, \Ff, (\Ff_t)_{t\geq 0}, \PP)$.
The SABR process $(X,Y)$ takes values in the state space $D= [0,\infty)\times(0,\infty)$, and describes the dynamics of a forward rate $X$ with stochastic volatility $Y$ and with the initial values $x_0>0$ and $y_0>0$.  
The constant elasticity of variance (CEV) parameter $\beta$ determines the general shape of the volatility 
smile and the parameter $\nu$ (often denoted by $\alpha$) governs the volatility the stochastic volatility.
The first pricing formula for the SABR model (the so-called Hagan formula) proposed 
in \cite{ManagingSmileRisk,HLW}, is based on an expansion of the Black-Scholes implied volatility for an asset driven by \eqref{eq:SABRSDE}. 
This tractable and easy-to-implement asymptotic expansion of the implied volatility made calibration to market data easier. This, and the model's ability to capture the shape and dynamics (when the current value $X_0=x$ of the asset changes) of the volatility smile observed in the market are   
virtues of the SABR model, which soon became a benchmark in interest rates derivatives markets~\cite{AndreasenHuge, AntonovFree, BallandTran, Rebonato}. 
The Hagan expansion is only accurate when the expansion parameter is small relative to the strike, that is when time to maturity or the volatility of volatility $\nu$ is sufficiently small. 
For low strike options such as in low interest rate and high volatility environments, much like the ones we are facing today, this formula can yield a negative density function for the process~$X$ in~\eqref{eq:SABRSDE}, which leads to arbitrage opportunities.
Therefore, as the problem of negative densities and arbitrage became more prevalent, it has been addressed
for example in \cite{Antonov,AntonovWings,BallandTran,Doust,HaganArbFreeSABR} by different approaches, some suggesting modifications of the SABR model or its implied volatility expansion. 
The attempt of suggesting suitable modifications to the original model is an intricate challenge since the Hagan expansion is deeply embedded in the market and fits market prices closely in moderate interest rate environments. 
Any model that deviates from its prices in those regimes may be deemed uncompetitive. This makes such pricing techniques desirable, which are applicable to the original model in all market environments.
There exist several refinements to this asymptotic formula:
in~\cite{Obloj} a correction the leading order term is proposed, 
and \cite{Paulot} provides a second-order term. 
In the uncorrelated case $\rho=0$ the exact density has been derived 
for the absolutely continuous part of the distribution of~$X$ on $(0,\infty)$ in~\cite{Antonov, FordeZhang, Islah}
and the correlated case was approximated by a mimicking model. 
However, it seems that these refinements have not fully resolved the arbitrage issue near the origin.
Recent results 
\cite{Doust, SABRMassZero, SABRMassZero2} focus on the singular part of the distribution and suggest an explanation for the irregularities appearing at interest rates near zero; and \cite{SABRMassZero, SABRMassZero2} provides a means to regularize Hagan's asymptotic formula at low strikes for specific parameter configurations, based on tail asymptotics derived in \cite{DMHJ,GulisashviliMass}.\\

We propose here a numerical pricing method for the (original) SABR model \eqref{eq:SABRSDE}
with rather mild assumptions on the parameters. It is consistently applicable in all market environments and allows for the derivation of convergence rates for the numerical approximations of option prices.
The most popular numerical approximation methods which were considered so far for the SABR model (or closely related models) fall into the following classes: probabilistic methods--- comprising of path simulation of the process combined with suitable (quasi-) Monte Carlo approximation---were considered in \cite{ChenOsterleeWeide} for the SABR model and \cite{Alfonsi2005, Alfonsi2007, ChassagneuxJacquierMihaylov, CoxHutzenthalerJentzen, HutzenthalerJentzenNoll} for related models. Furthermore in \cite{ChenOsterleeWeide} some difficulties of Euler methods in the context of SABR are discussed.
Splitting methods---where the infinitesimal generator of the process is decomposed into suitable operators for which the pricing equations can be computed more efficiently---provide a powerful tool in terms of computation efficiency for sufficiently regular processes. Such methods are considered in \cite{BayerGatheralKarlsmark} for a model closely related to SABR (see also \cite{BayerFrizLoeffen}), and in \cite{SemigroupPOVSplittingSchemes} for a large class of models. However, the applicability of corresponding convergence results to the SABR model itself is not fully resolved. Among fully deterministic PDE methods are most notably finite difference methods, which were considered in \cite{AndreasenHuge, LeFlochKennedy} for the modification \cite{HaganArbFreeSABR} of the SABR model \eqref{eq:SABRSDE} and
finite element methods, which were described in the context of mathematical finance in \cite{WilmottHowisonDewynne}. In the recent textbook \cite{ReichmannComputationalMethods} finite element methods have been applied to a large class of financial models---including the closely related process \eqref{eq:CEVprocess}---and provide a robust and flexible framework to handle the stochastic finesses of these models. In spite of this, finite element approximation methods did not appear in the context of the SABR model so far in the corresponding literature.  
For a broad review of simulation schemes used for the SABR model \eqref{eq:SABRSDE} in specific parameter regimes, see also \cite{Lord} and the references therein.
\\

Standard theory provides convergence of the above methods if the considered model satisfies certain (method-specific) regularity conditions. However in the case of the SABR model, obtaining convergence rates is non-standard for these methods: The degeneracy of the SABR Kolmogorov equation at the origin violates the assumptions needed in conventional finite difference methods and---for a range of parameters---also those of ad-hoc (i.e. unweighted) finite element methods.
Path simulation of the SABR process also requires non-standard techniques due to the degeneracy of the diffusion  \eqref{eq:SABRSDE} at zero. 
Nonstandard techniques often become necessary for the numerical simulation of a stochastic differential equation, 
when the drift and diffusion ($b$ and $\sigma$)  
do not satisfy the global Lipschitz condition
\begin{equation}\label{eq:GlobalLipschitzCond}
|b(x)-b(y)|+|\sigma(x)-\sigma(y)|\leq C |x-y|,
\end{equation}
for $x,y \in \R^n$ and a constant $C>0$ independent from $x$ and $y$, cf. \cite{ReichmannSchwab}. 
The degeneracy of the SABR model \eqref{eq:SABRSDE} at $X=0$, originates from failure of condition \eqref{eq:GlobalLipschitzCond} for the CEV process $\widetilde X$, described for parameters $\alpha>0$, and $\beta\in [0,1]$ by the equation
\begin{equation}\label{eq:CEVprocess}
\begin{array}{rll}
\D \widetilde{X}_t=\alpha \widetilde{X}_t^{\beta}\D W_t \qquad & \widetilde X_0 = \widetilde x_0> 0,\quad  0\leq t\leq T<\infty.
\end{array}
\end{equation}
Although the exact distribution of the CEV process is available \cite{LindsayBrecher}, simulation of the full SABR model based on it can in many cases become involved and expensive. In fact, exact formulas decomposing the SABR-distribution into a CEV part and a volatility part are only available in restricted parameter regimes, see \cite{Antonov, FordePogudin, Islah} for the absolutely continuous part and \cite{SABRMassZero, SABRMassZero2} for the singular part of the distribution.

A simple space transformation (see \eqref{eq:Spacetrafo} below)
makes some numerical approximation results for the CIR model (the perhaps most well-understood degenerate diffusion) applicable to certain parameter regimes of the SABR process. The CIR process
\begin{equation}\label{eq:CIRprocess}
\begin{array}{rll}
&\D S_t  = \left(\delta-\gamma S_t \right) \D t + a \sqrt{S_t} \D W_t, \quad  &S_0 = s_0> 0,\quad  0\leq t\leq T<\infty,
\end{array}
\end{equation}
with $a>0, \delta \geq 0$ indeed reduces for the parameters $\gamma=0$, $a=2$ to a squared Bessel process with dimension $\delta$ on the positive real line, the connection to CEV is then made via 
\begin{equation}\label{eq:Spacetrafo}
\begin{array}{rll}
\varphi:\R_{\geq0}&\longrightarrow\R_{\geq0}\\
s&\longmapsto \frac{1}{1-\delta/2} \  s^{1-\frac{\delta}{2}}, \quad \textrm{where} \quad 
\beta=\frac{1-\delta}{2-\delta}, \quad \textrm{for}\quad \delta\neq 2,
\end{array}
\end{equation}
that is, assuming absorbing boundary conditions at zero,  the law of $S$ in \eqref{eq:CIRprocess} (for the parameters $\gamma=0$, $a=2$) under the space transformation $\varphi$ in \eqref{eq:Spacetrafo} coincides with the law of $\widetilde {X}$ in \eqref{eq:CEVprocess}.
\\
Recent results in this direction, exploring probabilistic approximation methods for diffusions where the global Lipschitz continuity \eqref{eq:GlobalLipschitzCond} is violated, can be found for example in \cite{ChassagneuxJacquierMihaylov,CoxHutzenthalerJentzen} and \cite{HutzenthalerJentzenNoll}, see also \cite{AndersenCIR, DereichNeuenkirchSzpruch} and the references therein.
Establishing strong convergence rates in the case $2 \delta < a^{2}$ where the boundary is accessible as in \cite{HutzenthalerJentzenNoll}, 
is of particular difficulty, as this case renders coefficients of the SDE \eqref{eq:CIRprocess} neither globally, nor locally Lipschitz continuous on the state space. 
Yet, these convergence results do not cover the parameter range of the SABR model.
Further approximation schemes are presented in \cite{Alfonsi2005, Alfonsi2007} which apply to CIR processes with accessible boundary and both strong and weak convergence for the approximation are studied. The weak error analysis 
of Talay and Tubaro \cite{TalayTubaro}
yields second order convergence of the schemes proposed in \cite{Alfonsi2005, Alfonsi2007}, which covers the parameters $0<\beta<\frac{1}{2}$ but the results do not directly carry over to the case  $\frac{1}{2}\leq \beta <1$.\\ 

Here, we turn to a fully deterministic numerical method based on discretizations of the Kolmogorov partial differential equations, using finite elements. 
We derive the appropriate analytic setup to handle the degeneracy of the model at the origin. That is, we construct a suitable evolution triple of Sobolev spaces
with singular weights on which well-posedness of the variational formulation of the SABR-pricing equations holds. The proposed method for space-discretization is based on the Dirichlet form corresponding to the SABR stochastic differential equation. Specifically, using the Dirichlet form we recast the Kolmogorov pricing equations in weak (variational) form and show the so-called \emph{well posedness} of the latter.
We use the weighted multiresolution (wavelet) Galerkin discretization of \cite{BeuchlerSchneiderSchwab} in the state space 
to approximate variational solutions of the SABR-Kolmogorov pricing equations for financial contracts (vanilla and barrier options).
For the time discretization of the semigroup generated by the process \eqref{eq:SABRSDE} we propose a $\theta$-scheme. We derive approximation estimates tailored to our weighted setup, measuring the error between the true solution of the pricing equations and their projection to the discretization spaces. Based on these, we conclude error estimates akin to \cite{PetersdorffSchwab} for our fully discrete scheme. Under appropriate regularity assumptions on the payoff we obtain the full convergence rate for our finite element approximation.
The advantage of the presented method is that it allows for a consistent pricing with very mild parameter assumptions on the SABR process and it is robustly applicable for moderate interest rate environments as well as in the current low interest rate regimes. Furthermore, the proposed discretization can be applied to compute prices of compound options or multi-period contracts without substantial modifications of the numerical methodology, cf. \cite{ReichmannSchwab}.\\
\smallskip\\
The article is organized as follows.
Section \ref{Sec:ProblemFormulation} is devoted to the formulation of the SABR pricing problem in the appropriate analytic setting and an outline of the general idea of the variational analysis underlying the proposed finite element method for the SABR model.
In Section \ref{Sec:ProblemFormulation} we introduce some notations. In Section \ref{Sec:SettingandVariationalFormulation} we introduce the SABR Dirichlet form and cast the variational formulation of the SABR pricing equation in a suitable setting. We then propose a Gelfand triplet of spaces for our finite element discretization, consisting of a space $\mathcal{V}$ of admissible functions (the domain of the SABR-Dirichlet form), its dual space $\mathcal{V}^*$ and a pivotal Hilbert space ${\Hh}$, containing $\mathcal{V}$. In Section \ref{Sec:WellPosednessSABR} we briefly recall some relevant existing results to prove well-posedness of the SABR-pricing problem on the triplet ${\V} \subset \Hh \subset {\V^*}$, and conclude the existence of a unique weak solution to the variational formulation of the Kolmogorov partial differential equations on these spaces. We furthermore derive in this section a non-symmetric Dirichlet form for the SABR model, thereby extending the results of \cite{DoeringHorvathTeichmann} on Dirichlet forms on SABR-type models to the non-symmetric case.
In Section \ref{Sec:Discretization} we present the finite element discretization of the weak solution of the equation examined in the previous sections.
Section \ref{Sec:SpaceDiscretization} is devoted to the space discretization, which is carried out through a spline wavelet discretization of spaces $\mathcal{V}$ and ${\Hh}$. 
We review the multiresolution spline wavelet analysis of \cite{ReichmannComputationalMethods, PetersdorffSchwab} (in the unweighted case) to discretize the volatility dimension.
The forward dimension (the CEV part) is more delicate, due to its degeneracy at zero. In this case we apply the weighted multiresolution norm equivalences, proven in \cite{BeuchlerSchneiderSchwab} which are suitable to this degeneracy.
We pass from the univariate case to the bivariate case by constructing tensor products of the discretized spaces in each dimension as outlined in \cite{ReichmannComputationalMethods}.
Finally, we specify the mass- and stiffness matrices involved in the space discretization. 
In Section \ref{Sec:TimeDiscretization} we present the fully discrete scheme by applying a $\theta$-scheme in the time-stepping. We follow \cite{PetersdorffSchwab}, to conclude that the stability of the $\theta$-scheme continues to hold in the present setting of weighted Sobolev spaces.
In Section \ref{Sec:ErrorEsimates} we derive error estimates for our finite element discretization. 
In Section \ref{Sec:ApproximationEstimates} approximation estimates of the projection to our discretization spaces are established based on multiresolution (weighted) norm equivalences. We cast our estimates under specification of different regularity assumptions on the solution of our pricing equations. We use these estimates in Section \ref{Sec:DiscretizationError} to derive convergence rates for our finite element discretization and conclude that under some regularity assumptions on the payoff, examined in the previous section, we obtain the full convergence rate. 
We remark here, that the approximation- and error estimates presented in Section \ref{Sec:ErrorEsimates} for the SABR model 
readily yield the corresponding approximation- and error estimates for the CEV model as a direct corollary.
The well-posedness of the variational formulation of the CEV pricing equations has been studied in \cite{ReichmannComputationalMethods, ReichmannSchwab}, however to the best of our knowledge, a presentation of the full error analysis thereof was not available in the corresponding literature so far.


\section{Preliminaries and problem formulation}\label{Sec:ProblemFormulation}
\textbf{Preliminaries and Notations:} 
there exists a unique weak solution of the system~\eqref{eq:SABRSDE} which will be established via the associated martingale problem \cite[Theorem 21.7]{kallenberg2002foundations}, and by (pathwise) uniqueness of \eqref{eq:SABRSDE}, via \cite[Theorem 23.3]{kallenberg2002foundations}. 
Furthermore, the process $X$ in~\eqref{eq:SABRSDE} is a martingale whenever $\beta<1$~\cite[Theorem 5.1]{Hobson}, and for $\beta=1$ it is a martingale if and only if $\rho\leq 0$ \cite[Remark 2]{Jourdain}, see also \cite[Section 1, and Theorem 3.1]{MusielaLions}. 
For two norms on a space ${\V}$ the notation  $||\cdot||_{\V_1}\approx||\cdot||_{\V_2}$ indicates that the norms are equivalent on ${\V}$. Function spaces of bivariate functions will be denoted in italic ($\V, \Hh,\ldots$) and spaces of univariate functions by ($V, H,\ldots$) accordingly. For a domain $G\subset \R^2$ (resp. interval $I\subset \R$) we denote by $\Ll^1_{loc}(G)$ (resp. $L^1_{loc}(I)$) the locally integrable functions on $G$ (resp. $I$), and by $\C^{\infty}_0(G)$ (resp. $C^{\infty}_0(I)$) the smooth functions with compact support. Derivatives with respect to time will be denoted by $\dot{u}$, $\ddot{u}$, $\ldots$ accordingly to ease notation.

\subsection{Analytic setting for the SABR model}\label{Sec:SettingandVariationalFormulation}
In this section we establish the triplet $\V\subset {\Hh} \subset \V^*$ of spaces, tailored to the SABR process on which we cast the Kolmogorov pricing equations in variational form. We then proceed to show the well-posedness of these pricing equations on this triplet, i.e. that inequalities \eqref{eq:Defcontinuity} and \eqref{eq:DefGardinginequality} are fulfilled. We conclude the section by proving that the bilinear form resulting from the weak formulation of the pricing equations is indeed a non-symmetric Dirichlet form corresponding to the (unique) law of the SABR process.
Fix a time horizon $[0,T]$ 
and set $\widetilde Y=\log(Y)$, 
so that the SDE \eqref{eq:SABRSDE} becomes
\begin{equation}\label{eq:SABRSDELogVol}\begin{array}{rll}
&dX_t
=X_t^{\beta} \exp(\widetilde Y_t) dW_t
\qquad & X_0 = x_0> 0,\\ 
&d\tilde Y_t=\nu dZ_t- \tfrac{\nu^2}{2}dt \qquad & \tilde Y_0 = \log y_0,\ y_0>0\\
&d\langle W,Z\rangle_t=\rho dt.
\end{array}
\end{equation}
where we impose absorbing boundary conditions at $X=0$ to ensure martingality of the process. The solution to \eqref{eq:SABRPDELogVol} (and \eqref{eq:SABRSDE}) is then uniquely defined. Indeed, for the parameters $\beta \in [0.5,1]$ this is the only choice. 
We remark here that in some recent research outputs (for example in \cite{AntonovFree}) it is suggested to choose reflecting boundary conditions at the origin (for the regime $\beta \in [0,0.5)$) in order to accomodate to market conditions where interest rates can become negative. Our analysis can be easily adapted to that setting, but for brevity we restrict our presentation to the absorbing case.
The value at time $t\in J=(0,T)$ of a European-type contract on \eqref{eq:SABRSDELogVol} with payoff $u_0$ is\footnote{For notational simplicity we consider zero risk-free interest rates here. Furthermore, we make some standard technical assumptions on the payoff function $u_0$: In accordance with \cite[Equation (5.10) page 47]{ReichmannComputationalMethods} we assume $u_0$ to satisfy $u_0(0)=0$ and a polynomial growth condition (satisfied by vanilla contracts), and that $u_0\in\mathcal{H}$ (see Definition \ref{Def:HilbertSpaceSABR}) in accordance with \cite[Equations (2.10) and (2.12)]{PetersdorffSchwab}.}
\begin{equation}\label{eq:EuropeanOptionPrice}
u(t,z)=\mathbb{E}\left[u_0(Z_{\tau}^z)\right], \quad t\in J
\end{equation}
where $\tau:=(T-t)$ and $Z_{\tau}^z:=(X_{\tau},\widetilde Y_{\tau})$ is the process started at $z:=(x,\tilde y) \in \R_{\geq0}\times \R$, with $(x,\tilde y)=(X_t,\tilde Y_t)$ $\mathbb{P}$-a.s.  
Then for $u\in C^{1,2}(J;\R_{\geq 0}\times \R)\cap C^0(\bar{J};\R_{\geq 0}\times \R)$ the Kolmogorov pricing equation to \eqref{eq:SABRSDELogVol} is
\begin{equation}\label{eq:SABRPDELogVol}
\begin{array}{rlr}
&\dot{u}(t,z)-Au(t,z)=g(t,z) &\textrm{in} \ J\times \R_{\geq0}\times \R,\\
&u(0,z)=u_0(z) &\textrm{in} \ \R_{\geq 0}\times \R
\end{array}
\end{equation}
where 
$g\in L^2(J,\V)\cap H^1(J, \Hh)$
denotes a general forcing term (see \cite[Section 4]{ReichmannComputationalMethods}) with $\V^*$ as in \eqref{eq:GelfandTriplet} and where the infinitesimal generator $A$ of \eqref{eq:SABRSDELogVol} reads
\begin{equation}\label{eq:GeneratorLogscaleSABR}
Af=\frac{x^{2\beta}e^{2 \tilde y}}{2} \partial_{xx}f+\rho \nu x^{\beta}e^{\tilde y} \partial_{x \tilde y}f+\tfrac{1}{2}\nu^2 \partial_{\tilde y \tilde y}f
-\tfrac{1}{2} \nu^2 \partial_{\tilde y}f\qquad \textrm{for} \ f\in C_0^{2}(D) \subset \mathcal{D}(A).
\end{equation}
From now on we drop the tilde in the logarithmic volatility for notational convenience.
The operator $A$ in \eqref{eq:GeneratorLogscaleSABR} is a linear second order operator, degenerate (i.e. non-elliptic) at the boundary $\{(x,y):x=0,y>0\}$.
The domain $\mathcal{D}(A)$ 
of the operator $A$ is equipped with a norm $||\cdot||_{\mathcal{V}}$ (specified in Definition \ref{Def:WeightedSobolevSpaceVSABR} below) and the completion of $\mathcal{D}(A)$ under this norm will be denoted by $\mathcal{V}$.
Furthermore, we denote by ${\Hh}$ (specified in Definition \ref{Def:HilbertSpaceSABR}  below) a separable Hilbert space---henceforth \emph{the pivot space}---containing $\mathcal{V}$, such that $\mathcal{V}\hookrightarrow {\Hh}$ is a dense embedding. 
The inner product  $(\cdot,\cdot)_{\Hh}$ of ${\Hh}$ is extended to a duality pairing $(\cdot,\cdot)_{\mathcal{V}^*\times\mathcal{V}}$ on $\mathcal{V}^*\times\mathcal{V}$, where $\mathcal{V}^*$ denotes the dual space of $\mathcal{V}$, equipped with the dual norm $||\cdot||_{\mathcal{V}^*}$.
Identifying the Hilbert space ${\Hh}$ with its dual ${\Hh}^*$
we obtain the corresponding Gelfand-triplet
\begin{equation}\label{eq:GelfandTriplet1}
\mathcal{V}\hookrightarrow {\Hh}\cong {\Hh}^* \hookrightarrow \mathcal{V}^*, 
\end{equation}
where $\hookrightarrow$ denotes a dense embedding.\\ 
With view to the discretization, it is customary (cf. \cite{ReichmannComputationalMethods, PetersdorffSchwab, ReichmannSchwab}) to localize the spatial domain of the PDE ($\R_{\geq 0}\times \R$ in \eqref{eq:SABRPDELogVol}) at this point to a bounded domain $G\subset \R_{\geq 0}\times \R$ with Lipschitz boundary.
In what follows, all localization domains are rectanguar $G=[0,R_x)\times(-R_y,R_y)$ denoting the range of admissible values which can be taken by the price (and volatility) process.
\begin{remark}
For call and put options the error made by truncating the domain to $G\subset \R\times \R_{\geq 0}$ corresponds to approximating the option prices by a knock-out barrier options, up to the first hitting time of the boundary $\partial G$, see \cite[Sections 5.3 and 6]{ReichmannComputationalMethods}. The estimate in \cite[Theorem 5.3.1]{ReichmannComputationalMethods} can be directly applied to the volatility dimension $Y$, and yields that the truncated problem converges to the original problem exponentially fast. Furthermore, for the CEV process in the asset price the estimate serves as an upper bound by comparison between the CEV $(\beta\in[0,1))$  process and geometric Brownian motion $\beta=1$.
The probabilistic argument to estimate the localization error by a knock-out barrier option was suggested by Cont and Voltchkova in \cite{ContVoltchkova,ContVoltchkova2} even in a more general setting of L\'{e}vy models. 
Indeed, for the SABR model, the probability that the first hitting time of $R_x$ resp. $R_y$ occurs before $T$ converges to zero as $R_x,R_y\rightarrow \infty$. The lower boundary however cannot be truncated to any domain $(\epsilon, R_x)$ for a positive $\epsilon$ without possibly introducing a significant localization error. This is due to the fact that the SABR process accumulates a positive mass at zero for every $T>0$ whenever $\beta<1$. For details see \cite{SABRMassZero, SABRMassZero2}, where the mass at zero is calculated for several relevant parameter configurations.
\end{remark}
\begin{definition}\label{Def:HilbertSpaceSABR}
Let $G:=[0,R_x)\times(-R_y,R_y) \subset \R_+\times \R$, $R_x,R_y>0$ be an open subset. 
On $G$ we define  the weighted space
\begin{equation}\label{eq:HilbertSpaceNormSABR}
\begin{array}{ll}
 \Hh:=\Ll^2(G,x^{\mu/2})&=\{u:G\rightarrow \R \ \textrm{measurable } \ | 
\ ||u||_{\Ll^2(G,x^{\mu/2})}<\infty \}\\
\end{array}
\end{equation}
with 
\begin{align}\label{eq:muchoice}
\mu \in 
\begin{cases}
[-2\beta,0] \textrm{ for } \beta \in [0,\tfrac{1}{2})\\
[-1,1-2\beta]\textrm{ for } \beta \in [\tfrac{1}{2},1),
\end{cases}
\end{align}
where $||u||_{\Ll^2(G,x^{\mu/2})}^2:=(u,u)_{\Hh}$ for the bilinear form 
\begin{equation}\label{eq:HilbertScalarProduct}
(u,v)_{\Hh}:=\int_G u(x,y)v(x,y) \ x^{\mu} dxdy, \quad u,v \in V.
\end{equation}
\end{definition}
\begin{remark}
Note that
$(\Hh,(\cdot,\cdot)_{\Hh})$ is a Hilbert space, see Appendix \ref{Sec:WeightedSobolevSpaces}, Lemma \ref{Lem:ClosednessContinuousEmbedding} and \ref{Rem:L1loc}.
\end{remark}
\noindent A possible choice for the weight is $\mu=-\beta$ for any $\beta\in [0,1)$. 
Alternatively, one can distinguish the cases $\beta\in[\frac{1}{2},1)$ and $\beta\in[0,\frac{1}{2})$ and choose $\mu=-\beta$ for $\beta\in[\frac{1}{2},1)$, but $\mu=0$ for $\beta\in[0,\frac{1}{2})\cup\{1\}$. 
 Distinguishing the above parameter regimes has the advantage that we preserve the classical setting of an unweighted $\Ll^2(G)$-space for the parameters $\beta\in[0,\frac{1}{2})\cup\{1\}$.
The latter choice furthermore highlights that our analytic setting consistently extends the univariate  CEV case to the bivariate SABR case: see \cite[Section 4.5]{ReichmannComputationalMethods} and Remark \ref{Rem:CEVTriplet} for the analytic setting for CEV.

\begin{definition}\label{Def:WeightedSobolevSpaceVSABR}
Set $G:=[0,R_x)\times(-R_y,R_y) \subset \R_+\times \R$, $R_x,R_y>0$, the domain of interest. 
For the first coordinate we consider $L^2([0,R_x), x^{\mu/2})=\{u\in L^2([0,R_x))$ with $||x^{\mu/2} u ||_{L^2}<\infty\}$. 
On $L^2([0,R_x), x^{\mu/2})$ we define the space 
\begin{equation*}
V_x:=\overline{C_0^{\infty}([0,R_x))}^{||\cdot||_{V_x}} \ \textrm{where}\ \  
||u||_{{V_x}}^2:=||x^{\beta+\mu/2}\partial_x u)||_{L^2(0,R_x)}^2+||x^{\mu/2} u ||_{L^2(0,R_x)}^2, \ u\in C_0^{\infty}([0,R_x)).
\end{equation*}
For the second coordinate we consider on $L^2(-R_y,R_y)$ the space
\begin{equation*}
V_y:=H^1(-R_y,R_y), \ \textrm{where}\ \  ||u||_{{V_y}}^2:=||\partial_y u||_{L^2(-R_y,R_y)}^2+||u ||_{L^2(-R_y,R_y)}^2, \ \ u\in H^1(-R_y,R_y).
\end{equation*}
We define for the bivariate case
\begin{align}\label{eq:WeightedSobolevSpaceV}
&\V:=\left(V_x\otimes L^2(-R_y,R_y)\right)\bigcap\left(L^2([0,R_x), x^{\mu/2})\otimes V_y\right).
\end{align}
The dual space will be denoted by $\V^*$ and equipped with the usual dual norm
\begin{equation}\label{eq:WeightedSobolevSpaceVDualNorm}
||v||_{V^*}=\sup_{u\in \V}\frac{(v,u)_{\V^* \times \V}}{||u||_{\V}}, \quad v \in {\V}^*. 
\end{equation}
\end{definition}
\begin{remark}
Note that the norm on the space \eqref{eq:WeightedSobolevSpaceV} is by construction\footnote{See Appendix \ref{Sec:TensorHilbertSpacesExplicit} for details.} equivalent to
\begin{equation}\label{eq:WeightedSobolevSpaceVNorm}
 ||u||_{\V}^2\approx||x^{\beta+\mu/2} \partial_xu ||_{\Ll^2(G)}^2+|| x^{\mu/2} \partial_yu||_{\Ll^2(G)}^2+|| x^{\mu/2} u ||_{\Ll^2(G)}^2,
\quad \textrm{for}\quad u\in \V.
\end{equation}
\end{remark}
\begin{lemma}\label{Lem:DenseEmbeddingSABR}
For any $\mu$ in \eqref{eq:muchoice} the spaces ${\Hh}$, ${\V}$ and ${\V}^*$ in Definitions \ref{Def:HilbertSpaceSABR} and \ref{Def:WeightedSobolevSpaceVSABR} form a Gelfand
triplet. In particular, the inclusion maps are dense embeddings\footnote{The analogous statement in the one-dimensional ($V_x$) part was presented in the CEV-analysis in \cite[Equation (5.33), page 56]{ReichmannComputationalMethods}.}
\begin{equation}\label{eq:GelfandTriplet}
\begin{array}{lcll}
{\V}\hookrightarrow& {\Hh} &\hookrightarrow {\V}^*,
\end{array}
\end{equation}
where ${\V}$ and ${\Hh}$ are specified in \eqref{eq:HilbertSpaceNormSABR} and \eqref{eq:WeightedSobolevSpaceV}.
\end{lemma}
\noindent The scalar product $( \cdot,\cdot )_{\Hh}$ can hence (by Lemma \ref{Lem:DenseEmbeddingSABR}) be extended to a dual pairing $( \cdot,\cdot )_{{\V} \times {\V}^*}$.
\begin{proof}
The first inclusion follows by construction (cf Definition \ref{Def:WeightedSobolevSpaceVSABR}) and the second by direct approximation in the weighted space $\Hh=\Ll^2(G,x^{\mu/2})$ by smooth functions with compact support.
\end{proof}
\noindent We extend the inner product $(\cdot,\cdot)_{\Hh}$ of the Hilbert space to the dual pairing $(\cdot,\cdot)_{{\V}^*\times {\V}}$ as described above\footnote{Note that for this the operator $A$ need not be self-adjoint nor needs the associated bilinear form be symmetric.}. Recalling that $A \in \mathcal{L}(\mathcal{V};\mathcal{V}^*)$ we apply the duality pairing $(\cdot,\cdot)_{\mathcal{V}^*\times\mathcal{V}}$: A bilinear form $a(\cdot,\cdot): \mathcal{V}\times  \mathcal{V}\rightarrow \mathbb{R}$ is thus
associated with the operator $A$ in \eqref{eq:GeneratorLogscaleSABR} by setting
\begin{equation}\label{eq:DefBilinearFormGeneral}
a(u,v):=-(Au,v)_{\mathcal{V}^*\times\mathcal{V}}, \quad u,v\in \mathcal{V},
\end{equation}
where the operator $A$ acts on $\V$ in the weak sense, see in \eqref{eq:LogvolBilinearformSABRWeighted} below. 
Hence, we define the SABR-bilinear form by the relation \eqref{eq:DefBilinearFormGeneral}, 
which therefore reads
\begin{align}\label{eq:LogvolBilinearformSABRWeighted}
\begin{split}
a(u,v)
=&\tfrac{1}{2}\int\int_{G} x^{2\beta+\mu}e^{2y}\  \partial_x u \ \partial_x v \ dxdy
+\tfrac{2 \beta +\mu}{2}\int\int_{G} x^{2 \beta +\mu -1}e^{2y}\  \partial_x u \ v \ dxdy\\
&+\rho \nu \int\int_{G} x^{\beta+\mu}e^{y}\  \partial_x u \ \partial_y v \ dxdy
+ \rho \nu \int\int_{G} x^{\beta+\mu}e^{y}\  \partial_x u \  v \ dxdy\\
&+\tfrac{\nu^2}{2}\int\int_{G} x^{\mu}\  \partial_y u \ \partial_y v \ dxdy
+\tfrac{\nu^2}{2}\int\int_{G} x^{\mu}\ \partial_y u \ v \ dxdy,\qquad u,v \in \V,
\end{split}
\end{align}
which is obtained from \eqref{eq:DefBilinearFormGeneral}---by the divergence theorem together with $\V\subset \Ll^1_{loc}(G)$---when $Au$, $u \in \V$ are interpreted as weak derivatives\footnote{Multiplying $-Au$ with $v\in C_0^{\infty}$ and integrating gives \eqref{eq:DefBilinearFormGeneral}, partial integration (cf. \eqref{eq:WeakDerivativeTime}) then yields \eqref{eq:LogvolBilinearformSABRWeighted}.
}.
With the spaces $\V\subset\Hh\subset\V^*$ and the bilinear form \eqref{eq:DefBilinearFormGeneral} at hand, we can formulate the variational (or weak) framework corresponding to \eqref{eq:SABRPDELogVol}.
The motivation for passing to the weak formulation is that for degenerate equations (such as \eqref{eq:SABRPDELogVol}) it is often not possible to find a classical solution $u \in C^{1,2}(J,\R^2)\cap C^{0}(\bar{J},\R^2)$ to the original problem \eqref{eq:SABRPDELogVol}. The variational reformulation \eqref{eq:SABRVariationalFormulation} problem may then still permit a (weak) solution $u$ with less regularity. Whenever a (weak) solution of the variational problem is  sufficiently smooth, it coincides with the solution of the corresponding original problem.
\begin{definition}[Variational formulation of the SABR pricing equation]\label{Def:SABRVariationalFormulation} Let ${\V},{\V}^*$ and ${\Hh}$ (resp. $\Ll^2(G,x^{\mu/2})$) be as in Definitions \ref{Def:WeightedSobolevSpaceVSABR} and \ref{Def:HilbertSpaceSABR}, and let the bilinear form $a(\cdot,\cdot)$ on ${\V}$ be as in \eqref{eq:LogvolBilinearformSABRWeighted}. Furthermore, let $u_0\in \Hh$ (resp. $u_0\in \Ll^2(G,x^{\mu/2})$) and consider for a $T>0$ the finite interval $J=(0,T)$. Then the variational formulation of the SABR pricing problem reads as follows: Find $u \in L^2(J;{\V})\cap H^1(J;{\V}^*)$, such that 
$u(0)=u_0$,  and for every $v\in {\V}$, $\varphi \in C_0^{\infty}(J)$
\begin{equation}\label{eq:SABRVariationalFormulation}
-\int_{J}(u(t),v)_H\ \dot{\varphi}(t) dt+\int_{J}a(u(t),v)\ \varphi(t) dt=\int_{J}(g(t),v)_{{\V}^*\times {\V}} \ \varphi(t)dt.
\end{equation}
\end{definition}
\noindent The time derivative of a function $u$ in the appropriate Bochner space is understood in the weak sense: For $u\in L^2(J;{\V})$, its weak derivative in $\dot{u}\in L^2(J,{\V}^*)\cap H^1(J;{\V}^*)$ is defined by the relation
\begin{equation}\label{eq:WeakDerivativeTime}
\int_J(\dot{u}(t),v)_{{\V}^*\times {\V}}\ \varphi (t)dt=-\int_J(u(t),v)_{{\V}^*\times {\V}}\ \dot{\varphi}(t)dt, 
\end{equation}
see \cite[Sections 2.1 and 3.1]{ReichmannComputationalMethods} for definitions and properties of Bochner spaces.
\begin{remark}[The CEV case: $\nu=0$]\label{Rem:CEVTriplet}
In \cite{ReichmannComputationalMethods,ReichmannSchwab} a corresponding analytic setup is studied for the univariate case: For the CEV model the Gelfand triplet $V\subset H\cong H^* \subset V$  in \cite{ReichmannComputationalMethods,ReichmannSchwab} consists of the weighed spaces 
\begin{equation*}
H:=L^2((0,R),x^{\mu/2})
\end{equation*}
and 
\begin{equation*}
V:=\overline{C_0^{\infty}([0,R))}^{||\cdot||_V} \quad \textrm{with} \quad 
||u||_{V}^2:=||x^{\beta+\mu/2}\partial_xu)||_{L^2(0,R)}^2+||x^{\mu/2} u ||_{L^2(0,R)}^2
\end{equation*}
such as its dual $V^*$, where the parameter $\mu \in [\max\{-1,-2\beta\},1-2\beta]$ is chosen as in Definitions \ref{Def:HilbertSpaceSABR}, and \ref{Def:WeightedSobolevSpaceVSABR}. 
Indeed, setting $\nu=0$, the SABR process \eqref{eq:SABRSDE} (resp. \eqref{eq:SABRSDELogVol}) with trivial volatility process reduces to the CEV model, and the state space $G$ reduces to $[0,R_x)\subset \R$. Accordingly, for $\nu=0$ the spaces $\Hh$, $\V$, $\V^*$ in Definitions \ref{Def:HilbertSpaceSABR} and \ref{Def:WeightedSobolevSpaceVSABR} coincide with the spaces $V$, $H$ and $V^*$ above. Also the SABR bilinear form \eqref{eq:LogvolBilinearformSABRWeighted} reduces to the corresponding CEV bilinear form. Hence, our analytic setup extends the univariate setup of the CEV model consistently to the bivariate SABR-case.
See \cite[Equation (21)]{ReichmannSchwab},\cite[Equations (4.30) and (4.33)]{ReichmannComputationalMethods} such as \cite[page 62]{ReichmannComputationalMethods} for the definitions\footnote{Note that $V,H,V^*$ are presented here in a form which is adjusted to our current notation.} of $V, H, V^*$ and \cite[Equation (4.28)]{ReichmannComputationalMethods} for the CEV bilinear form.
\end{remark}

\subsection{Well-posedness of the variational pricing equations and SABR Dirichlet forms}\label{Sec:WellPosednessSABR}
In this section we show that the triplet of spaces ${\V}\subset {\Hh}\cong {\Hh}^*\subset {\V}^*$  (Definitions \ref{Def:HilbertSpaceSABR} and \ref{Def:WeightedSobolevSpaceVSABR})
is tailored to the degeneracy of the infinitesimal generator \eqref{eq:GeneratorLogscaleSABR} at zero: in the sense
that in this setting (as a consequence of Theorem  \ref{Th:WellPosednessGeneral} and well-posedness cf. Theorem  \ref{Th:UniqueSolSABR}), the variational formulation of the SABR pricing equation has a unique solution in $\V$ and hence the family of option prices $P_t=E[u_0(Z_t)]$, $t\geq 0$ is a strongly continuous contraction semigroup on the Hilbert space $\Hh$, cf. Theorems \ref{Th:RegularityAprioriEstimatesGeneral} and \ref{Th:UniqueSolSABR}.\\ 
Key result in this section is the well-posedness of the SABR variational equations,  established in Theorem \ref{Th:UniqueSolSABR}. 
Furthermore, we show that the SABR-bilinear form \eqref{eq:LogvolBilinearformSABRWeighted} is a (non-symmetric) Dirichlet form with domain $\V$ on the Hilbert space $\Hh$, cf. Theorem \ref{Th:SABRDirichletForm}. The latter extends the results of \cite{DoeringHorvathTeichmann} on SABR-Dirichlet forms. We start by briefly recalling some concepts and results used in this section.
\begin{definition}[Continuity]\label{Def:continuity}
The form $a(\cdot, \cdot)$ is called \emph{continuous on ${\V}$} if there exists a $0<C_1<\infty$ such that 
\begin{equation}\label{eq:Defcontinuity}
\forall u,v \in {\V}: \quad |a(u,v)|\leq C_1 ||u||_{\V} ||v||_{\V}. 
\end{equation}
\end{definition}
\begin{definition}[Coercivity, G\aa{}rding inequality]
The form \eqref{eq:DefBilinearFormGeneral} is said to \emph{satisfy the G\aa{}rding inequality on  ${\V}$}, if there exists a constant $C_3\geq0$ such that 
\begin{equation}\label{eq:DefGardinginequality}
\forall u \in {\V} \ : \ a(u,u) \geq C_2 ||u||^2_{\V}- C_3||u||^2_{\Hh}. 
\end{equation}
If \eqref{eq:DefGardinginequality} holds with $C_3=0$, the form $a(\cdot,\cdot)$ is \emph{coercive}\footnote{It is standard (see for example\cite[Remark 2.4]{PetersdorffSchwab}) that the (weaker) G\aa{}rding inequality can be reduced to the coercivity property by the substitution
$v:=\E^{-C_3 t}u$. In case \eqref{eq:DefGardinginequality} is fulfilled for the operator $A$ at $u$, then \eqref{eq:DefGardinginequality} with $C_3=0$ is fulfilled at $v$. Then the operator $A+C_3 I$ is coercive and solves the related problem
\begin{equation*}
\dot{v}(t,z)+ (A+ C_3 I)v(t,z)= \E^{-C_3 t}g(t,z)\quad \textrm{in} \quad t\in J, z\in \R^2.
\end{equation*}
} and the equivalence $a(\cdot,\cdot)\approx ||\cdot||^2_{\V}$ holds.
\end{definition}
\begin{theorem}\label{Th:WellPosednessGeneral}
Let ${\V}$ and ${\Hh}$ be separable Hilbert spaces with a continuous dense embedding ${\V}\hookrightarrow {\Hh}$. Furthermore, let $a:{\V}\times {\V}\rightarrow \R$ be a bilinear form satisfying the inequalities \eqref{eq:Defcontinuity} and \eqref{eq:DefGardinginequality}.
Then the corresponding variational parabolic problem (recall Definition \ref{Def:SABRVariationalFormulation})
has a unique solution in $L^2(J;{\V})\cap H^1(J;{\V}^*)$.
\end{theorem}
\begin{proof}
See \cite[Theorem 4.1]{LionsMagens} for a proof.
\end{proof}
\begin{theorem}\label{Th:RegularityAprioriEstimatesGeneral} Consider a bilinear form $a(\cdot,\cdot):{\V}\times {\V}\rightarrow \R$ associated with an $A\in \mathcal{L}({\V},{\V}^*)$ via \eqref{eq:DefBilinearFormGeneral}. If $a(\cdot,\cdot)$ satisfies the properties \eqref{eq:Defcontinuity} and  \eqref{eq:DefGardinginequality}, then  
$-A$ is the infinitesimal generator of a bounded analytic $C^0$-semigroup $(P_t)_{t\geq 0}$ in ${\V}^*$.
In this case, for given $u_0 \in {\Hh}$ and $g \in L^2(J;{\V}^*)$, the unique\footnote{By Theorem \ref{Th:WellPosednessGeneral} above.} variational solution of the corresponding equation can be represented as
\begin{equation}\label{eq:SolutionRepresentation1}
u(t)=P_tu_0+\int_{0}^{t}P_{t-s}g(s)ds.
\end{equation}
\end{theorem}
\begin{proof} See \cite[Theorem 2.3.]{MatacheSchwab}, \cite[Section 2. Equation (2.13) and Remark 2.1]{PetersdorffSchwab} and also \cite{LionsMagens}.
\end{proof}
\noindent We now formulate the main theorem in this section:
\begin{theorem}[Well-posedness of the SABR pricing equation]\label{Th:UniqueSolSABR}
For every configuration $(\beta,|\rho|,\nu) \in [0,1]\times[0,1] \times \R_+$ of the SABR parameters, which  satisfy the condition $|\rho|\nu^2<2$ and for any $x_0,y_0>0$, the variational formulation \eqref{eq:SABRVariationalFormulation} of the pricing equation \eqref{eq:SABRPDELogVol} corresponding to the SABR model \eqref{eq:SABRPDELogVol} admits
a unique solution $u \in L^2(J,\V)\cap H^1(J,\V^*)$ for any forcing term $g\in  L^2(J,\mathcal{V}^\ast)$ and any $u_0$ in $\mathcal{H}$.
The unique variational solution of the pricing equation can be represented as
\begin{equation}\label{eq:SolutionRepresentation}
u(t,z)=P_tu_0(z)+ \int_0^tP_{t-s}g(s)ds, \qquad t\geq 0, z\in G
\end{equation}
for a strongly continuous semigroup $(P_t)_{t\geq0}$ on $\Hh$ with the infinitesimal generator $A$ in \eqref{eq:GeneratorLogscaleSABR}.
\end{theorem}
\begin{proof}
This is a direct consequence of Theorem \ref{Th:WellPosednessGeneral} applied to Lemmas \ref{Th1:WellPosednessSABR} and \ref{Th2:WellPosednessSABR} below, which establish continuity \eqref{eq:Defcontinuity} and the G\aa{}rding inequality \eqref{eq:DefGardinginequality} for the SABR Dirichlet form \eqref{eq:LogvolBilinearformSABRWeighted} on this triplet.
\end{proof}
\begin{lemma}\label{Th1:WellPosednessSABR}
The bilinear form \eqref{eq:LogvolBilinearformSABRWeighted} is continuous: There exists a constant $C_1>0$ such that
\begin{align}
|a(u,v)|\leq C_1 \ ||u||_{\V} \ ||v||_{\V}, \quad \textrm{for all} \quad u,v \in \V \label{eq:WellPosednessSABRContinuity}.
\end{align}
\end{lemma}
\begin{proof}
The continuity statement \eqref{eq:WellPosednessSABRContinuity} is a direct consequence of the following six estimates, each of which corresponds to a component of $a(\cdot,\cdot)$ in \eqref{eq:LogvolBilinearformSABRWeighted}:\\
\begin{enumerate}
\item[(1)]
$\frac{1}{2}\int\int_{G} x^{2\beta+\mu}e^{2y} (\partial_x u) (\partial_x v) dxdy
\leq \frac{1}{2}\left(||x^{\beta+\mu/2}e^{y} \partial_x u||_{\Ll^2(G)}^2+||x^{\beta+\mu/2}e^{y} \partial_x v||_{\Ll^2(G)}^2\right),$\\
\smallskip
\item[(2)] By the Cauchy-Schwarz inequality,
\begin{equation*}\begin{array}{ll}
&\frac{2 \beta +\mu}{2}\int\int_{G} x^{2 \beta +\mu -1}e^{2y} (\partial_x u)  v dxdy 
\\&\leq \frac{2 \beta +\mu}{2}\left(\int\int_{G} x^{2 \beta +\mu -2}e^{2y}   v^2 dxdy \right)^{1/2}
\left(\int\int_{G} x^{2 \beta +\mu}e^{2y}  (\partial_x u)^2 dxdy\right)^{1/2},
\end{array}\end{equation*}
and an upper bound for the latter is derived from Hardy's inequality \cite[(5.56) p. 54]{ReichmannComputationalMethods}\footnote{We use Hardy's inequality as in \cite[(5.56) p. 54]{ReichmannComputationalMethods}, setting $\varepsilon=2 \beta+\mu-1$ and $C=\frac{2\beta+\mu}{2}$:
$C \int \int x^{\varepsilon-2}v^2(x,y) dx \ e^y dy \leq C \int \frac{1}{|\varepsilon-1|} \int x^{\varepsilon}|\partial_x v(x,y)|^2 dx e^{y}dy$.}:
\begin{equation*}\begin{array}{ll}
&\leq \frac{2 \beta +\mu}{2}
\frac{2}{|2 \beta +\mu -1|}
||x^{ \beta +\mu/2}e^{y} \partial_x v||_{\Ll^2(G)}|| x^{ \beta +\mu/2}e^{y} \partial_x u||_{\Ll^2(G)}
\\&\leq \frac{2 \beta +\mu}{2}
\frac{2}{|2 \beta +\mu -1|}
\left(||x^{ \beta +\mu/2}e^{y} \partial_x v||_{\Ll^2(G)}^2 +||x^{ \beta +\mu/2}e^{y} \partial_x u||_{\Ll^2(G)}^2\right)
\end{array}\end{equation*}
\item[(3)]
$\rho \nu \int\int_{G} x^{\beta+\mu}e^{y} (\partial_x u) (\partial_y v) dxdy 
\leq |\rho \nu| \ \left(||x^{\beta+\mu/2}e^{y} \partial_x u||_{\Ll^2(G)}^2+ \frac{\nu^2}{2}||x^{\mu/2} \partial_y v||_{\Ll^2(G)}^2\right)$,\\
\smallskip
\item[(4)]
$\rho \nu \int\int_{G} x^{\beta+\mu}e^{y} (\partial_x u) v dxdy 
\leq |\rho \nu| \left( ||x^{\beta+\mu/2}e^{y} \partial_x u||_{\Ll^2(G)}^2+ \frac{\nu^2}{2}||x^{\mu/2}v||_{\Ll^2(G)}^2\right),\\
\smallskip
$
\item[(5)]
$\frac{\nu^2}{2}\int\int_{G} x^{\mu} (\partial_y u)(\partial_y v) dxdy
\leq \frac{\nu^2}{2}\left(||x^{\mu/2} \partial_y u||_{\Ll^2(G)}^2+ ||x^{\mu/2} \partial_y v||^2_{\Ll^2(G)}\right)$,\\
\smallskip
\item[(6)] 
$\frac{\nu^2}{2}\int\int_{G} x^{\mu} (\partial_y u) v dxdy \leq 0
\leq \frac{\nu^2}{2}\left(||x^{\mu/2} \partial_y u||_{\Ll^2(G)}^2+ ||x^{\mu/2} v||^2_{\Ll^2(G)}\right).$
\end{enumerate}
\end{proof}
\begin{remark}\label{Rem:ContinuityNonTruncated}
The proof of Lemma \ref{Th1:WellPosednessSABR} reveals that analogous estimates are valid if the domain $G$ is the whole (not-truncated) state space $\R\times \R_{\geq 0}$. In the next lemma, an analogous statement holds true if in estimate (2) and (6) integration by parts is valid with vanishing boundary terms.
\end{remark}
\begin{lemma}\label{Th2:WellPosednessSABR}
The bilinear form (\ref{eq:LogvolBilinearformSABRWeighted}) satisfies the G\aa{}rding inequality, i.e. there exist constants $C_2>0$ and $C_3\geq 0$ such that
\begin{align}
a(u,u)\geq C_2 \ ||u||_{\V}^2-C_3||u||_{\Hh}^2,  \quad \textrm{for all} \quad u \in {\V}\label{eq:WellPosednessSABRGarding}.
\end{align}
\end{lemma}
\begin{proof}
The G\aa{}rding inequality \eqref{eq:WellPosednessSABRGarding} is obtained from the following estimates:
\begin{enumerate}
\item[(1)]$\frac{1}{2}\int\int_{G} x^{2\beta+\mu}e^{2y} \partial_x u \partial_x u dxdy=\frac{1}{2}||x^{\beta+\mu/2}e^{y} \partial_x u||_{\Ll^2(G)}^2$\\
\smallskip
\item[(2)] Using $(\partial_x u)u=\frac{1}{2}\partial_x(u^2)$ and integration by parts in the second term of \eqref{eq:LogvolBilinearformSABRWeighted} yields 
\begin{equation*}\begin{array}{ll}
\frac{2 \beta +\mu}{2}\int\int_{G} x^{2 \beta +\mu -1}e^{2y} \frac{1}{2}\partial_x (u^2) dxdy 
=-\frac{2 \beta +\mu}{2}(2 \beta +\mu -1)\int\int_{G} x^{2 \beta +\mu -2}e^{2y} u^2 dxdy
\end{array}\end{equation*}
The last term is non-negative if and only if $\mu \in [-2\beta,1-2\beta]$, which is satisfied by \eqref{eq:muchoice}. \\
\smallskip
\item[(3)]
 $\rho \nu \int\int_{G} x^{\beta+\mu}e^{y} (\partial_x u) (\partial_y u) dxdy 
\geq \frac{-|\rho| \nu^3 }{4}\left( \frac{1}{\delta}||x^{\beta+\mu/2}e^{y} \partial_x u||_{\Ll^2(G)}^2+\delta|| x^{\mu/2} \partial_y u||_{\Ll^2(G)}^2\right),\\
 $
for a constant $\delta>0$. \\
\smallskip
\item[(4)]
$\rho \nu \int\int_{G} x^{\beta+\mu}e^{y} (\partial_x u) (\partial_y u) dxdy 
\geq \frac{-|\rho| \nu^3 }{4}\left(\epsilon ||x^{\beta+\mu/2}e^{y} \partial_x u||_{\Ll^2(G)}^2+\tfrac{1}{\epsilon}|| x^{\mu/2} u||_{\Ll^2(G)}^2\right),$ \\           
for a constant $\epsilon>0$.\\
\smallskip
\item[(5)]$\frac{\nu^2}{2}\int\int_{G} x^{\mu} \partial_y u \partial_y u dxdy
=\frac{\nu^2}{2}||x^{\mu/2} \partial_y u||_{\Ll^2(G)}^2$,\\
\smallskip
\item[(6)]$ \frac{\nu^2}{2}\int\int_{G} x^{\mu} (\partial_y u) u dxdy=
\frac{\nu^2}{2}
\int \left(\int  u(\partial_y u)   dy \right) x^{\mu} dx
= \frac{\nu^2}{2} \int \left(-\int  u(\partial_y u)   dy \right) x^{\mu} dx
=0$ 
by integration by parts, the (Dirichlet) boundary conditions for $u\in C_0^{\infty}(G)$ and by the density of $C_0^{\infty}(G)$ in ${\V}$.
\end{enumerate}
Hence,
\begin{equation*}\begin{array}{ll}
a(u,u)
\geq \left(\tfrac{1}{2}-\tfrac{|\rho| \nu^3}{4 \delta}-\tfrac{|\rho| \nu^3 \epsilon}{4}\right) ||x^{\beta+\mu/2}e^{y} \partial_x u||_{\Ll^2}^2 +\left(\tfrac{\nu^2}{2}-\tfrac{|\rho| \nu^3 \delta}{4}\right)||x^{\mu/2} \partial_y u||_{\Ll^2}^2
-\frac{|\rho| \nu^3 }{4 \epsilon}|| x^{\mu/2} u||_{\Ll^2}^2, 
\end{array}
\end{equation*}
which yields the inequality \eqref{eq:WellPosednessSABRGarding}
\begin{equation*}\begin{array}{ll}
a(u,u)
&\geq C_2 \left(||x^{\beta+\mu/2}e^{y} \partial_x u||_{\Ll^2(G)}^2 +||x^{\mu/2} \partial_y u||_{\Ll^2(G)}^2+|| x^{\mu/2} u||_{\Ll^2(G)}^2\right)
-C_3|| x^{\mu/2} u||_{\Ll^2(G)}^2\\
&=C_2||u||_{\V}^2-C_3||u||_{\Hh}^2, 
\end{array}\end{equation*}
with $C_2:=\min\{\frac{\nu^2}{2}-\frac{|\rho| \nu^3 \delta}{4},\frac{1}{2}-\frac{|\rho| \nu^3}{4 \delta}-\frac{|\rho| \nu^3 \epsilon}{4}\}$ and $C_3:=C_2+\frac{|\rho| \nu^3 }{4 \epsilon}$.
\end{proof}
It remains to verify that the constants $\delta$ and $\epsilon$ can be chosen accordingly such that $C_2>0$.
\begin{lemma}\label{Th:ExistenceofConstants}
For every configuration $(\beta,|\rho|,\nu) \in [0,1]\times[0,1] \times \R_+$ of the SABR parameters, which satisfy the condition $|\rho|\nu^2<2$ there exist constants $\epsilon>0$ and $\delta>0$ such that  
$
C_2>0. 
$
\end{lemma}
\begin{remark}[Discussion of the parameter restrictions]\label{Rem:ParameterRestrictions}
Note that the case $\rho=0, \nu>0$ readily yields $C_2>0$ and for $\nu=0$ yields the CEV model, as discussed in Remark \ref{Rem:CEVTriplet}. Furthermore, the condition on the parameters in Lemma \ref{Th:ExistenceofConstants} is satisfied for any $|\rho| \in (0,1]$ for example if $0<\nu<\sqrt{2}$.
The latter condition on $\nu$ is fulfilled in most practical scenarios observed in the market as usual values of this parameter are well below $\sqrt{2}$: The volatility of volatility typically calibrates to values around $\nu=0.2;\ 0.4;\ 0.6$, see for example \cite{ManagingSmileRisk, Obloj, RebonatoMcKay}.
\end{remark}
\begin{proof}[Proof of Lemma \ref{Th:ExistenceofConstants}]
For any parameter configuration with $|\rho|\nu^2<2$ one can choose the constant $\delta$ in such a way that
\begin{equation}\label{Constant2}\begin{array}{ll}
\frac{2}{|\rho|\nu}>\delta>0, 
\end{array}\end{equation}
and the constants $\epsilon$ accordingly, such that
\begin{equation}\label{Constant3}\begin{array}{ll}
\frac{2}{\nu^3 |\rho|}-\frac{1}{\delta}>\epsilon>0. 
\end{array}\end{equation}
If the inequalities \eqref{Constant2} and \eqref{Constant3} are satisfied then
$C_2>0$ follows.
It remains to verify that $(\ref{Constant3})$ poses no contradiction to $(\ref{Constant2})$.
The bounds on $\delta$ are
\begin{equation}\label{Set}\begin{array}{ll}
\delta \in \left( \frac{|\rho|\nu^3}{2},\frac{2}{|\rho|\nu}\right).
\end{array}
\end{equation}
Indeed if $|\rho|\nu^2<2$, then the set in (\ref{Set}) is nonempty, and there exists an $\epsilon$ satisfying \eqref{Constant3}.
\end{proof}
\begin{theorem}[The non-symmetric SABR Dirichlet form]\label{Th:SABRDirichletForm}
Let the bilinear form $a(\cdot,\cdot)$ be as in \eqref{eq:LogvolBilinearformSABRWeighted} and its domain $\V$ as in  \eqref{eq:WeightedSobolevSpaceV}. Then the pair $(a(\cdot,\cdot),\V)$ is a (non-symmetric) Dirichlet form  
on the Hilbert space $(\Hh,(\cdot,\cdot)_{\Hh})$ in \eqref{eq:HilbertSpaceNormSABR}, for every parameter configuration $(\beta,|\rho|,\nu) \in [0,1]\times[0,1] \times \R_+$ with $|\rho|\nu^2<2$, and for any $\mu$ as in \eqref{eq:muchoice}. 
\end{theorem}
\begin{proof}
The crucial statement in the above theorem is that the pair $(a,\V)$ is a \emph{coercive closed form} on the Hilbert spaces $\Hh$ and $\Ll^2(G,x^{\mu/2})$ for any $\mu$ as in \eqref{eq:muchoice}, where $a$ denotes SABR-bilinear form \eqref{eq:LogvolBilinearformSABRWeighted} and $\V$ is the space \eqref{eq:WeightedSobolevSpaceV}. 
For this we first define an auxiliary symmetric bilinear form related to \eqref{eq:LogvolBilinearformSABRWeighted} as:
\begin{equation}\label{eq:auxiliarySymmetricForm}
\begin{array}{rl}
\cE(u,v)&:=\int \int_{G} x^{2\beta+\mu}e^{2y} \partial_x(u) \partial_x(v)\ dx\ dy +\int\int x^{\mu} \partial_y(u) \partial_y(v)\ dx \ dy\\
\mathcal{D}(\cE)&:=\C_0^{\infty}(G).
\end{array}
\end{equation}
It follows directly from \cite[Corollary 3.5]{RoecknerWielens} or \cite[Proposition 1]{BouleauDenis} that for any open $G\subset \mathbb{R}_{\geq 0}\times \R$, the
bilinear form \eqref{eq:auxiliarySymmetricForm}
is closable on the Hilbert space $\Hh=\Ll^2(G,x^{\mu/2})$. 
Then, closability of $(a,\C_0^{\infty}(G))$ in $\Hh$ and $\Ll^2(G,x^{\mu/2})$ is inherited from the closability of the auxiliary symmetric form $\cE$ in \eqref{eq:auxiliarySymmetricForm} via the equivalence of norms $a(\cdot,\cdot)\approx ||\cdot||_{\V}^2$, see \cite[Section 3, Proposition 3.5]{RoecknerMa} . The latter equivalence is a direct consequence of the continuity property \eqref{eq:WellPosednessSABRContinuity} and G\aa{}rding inequality \eqref{eq:WellPosednessSABRGarding}, which were proven for the form $a(\cdot,\cdot)$ on the triplets $\V \subset \Hh \subset \V^*$ such as $\V \subset \Ll^2(G,x^{\mu/2})\subset \V^*$ in Section \ref{Sec:WellPosednessSABR}.
Hence, in the inequality \eqref{eq:WellPosednessSABRContinuity} the norm $||\cdot||_{\V}$ on the right hand side can be replaced by $\left(a(\cdot,\cdot)\right)^{1/2}$, yielding the strong (resp. weak) sector condition (see \eqref{eq:weaksector} and Remark \ref{Rem:StrongSectorWeakSector}) for the SABR-bilinear form: There exists a $C_1>0$, such that for all $u,v\in \V$
\begin{equation*}\label{eq:SectorCondition}
\begin{array}{lll}
|a(u,v)|\leq C_1 \  \left(a(u,u)\right)^{1/2}\left(a(v,v)\right)^{1/2}&\textrm{resp.}&|a_1(u,v)|\leq C_1 \  \left(a_1(u,u)\right)^{1/2}\left(a_1(v,v))\right)^{1/2}
\end{array}
\end{equation*}
for $a_1(u,v):=a(u,v)+(u,v)_{\Hh}$,  $u,v\in \V$.
cf. \cite[Equations (2.4) resp. (2.3)]{RoecknerMa}. Hence, $(a(\cdot,\cdot),\V)$ is a \emph{coercive closed form} on $(\Hh,(\cdot,\cdot)_{\Hh})$ and on $(\Ll^2(G,x^{\mu/2}),(\cdot,\cdot)_{\Hh}))$, in the sense of \cite[Definition 2.4]{RoecknerMa}, see Section \ref{Def:CoerciveClosedForm} below.\\ 
Now for the coercive closed form $(a(\cdot,\cdot),\V)$ to be a Dirichlet form (cf. Definition \ref{Def:NonSymmDirichlet}) it remains to show that it is sub-Markovian (i.e. it satisfies contraction properties \eqref{eq:Contractionproperties}).
These contraction properties follow  via Theorem \ref{Th:EquivalenceContractionProperties} from the respective contraction properties of the unique (cf. Theorem \ref{Th:AssociatedSemigroup}) semigroups on $(\Hh,(\cdot,\cdot)_{\Hh})$ and $(\Ll^2(G,x^{\mu/2}),(\cdot,\cdot)_{\Hh}))$ associated to $(a(\cdot,\cdot),\V)$.\\
The contraction properties (sub-Markovianity) in Definition \ref{Def:CoerciveClosedForm} for the bilinear form \eqref{eq:LogvolBilinearformSABRWeighted} can also be shown directly:
For any $u\in \V$ it holds that $u^+\wedge 1\in \V$ (since derivatives are taken in the weak sense) and the the contraction properties are by non-negativity of the form equivalent to 
\begin{equation}\label{eq:Contractionproperties}
\begin{array}{ll}
&a(u+u^+\wedge 1,u-u^+\wedge 1)\geq 0 \qquad \textrm{ if and only if } \qquad a(u^+\wedge 1,u-u^+\wedge 1)\geq 0,\\
&a(u-u^+\wedge 1,u+u^+\wedge 1)\geq 0 \qquad \textrm{ if and only if } \qquad a(u-u^+\wedge 1,u^+\wedge 1)\geq 0.
\end{array}
\end{equation}
Since the functions $u^+\wedge 1$ and $u-u^+\wedge 1$ have disjoint supports, the assertion $a(u^+\wedge 1,u-u^+\wedge 1)= 0$ follows directly from the construction \eqref{eq:LogvolBilinearformSABRWeighted} of the bilinear form, yielding sub-Markovianity. Hence, the form $a(u^+\wedge 1,u-u^+\wedge 1)$ is a non-symmetric (as $a(u,v)\neq	a(v,u)$ in general) Dirichlet form.
\end{proof}
\noindent For the sake of completeness, we included in Appendix \ref{Sec:NonSymmDirichlet} the properties of non-symmetric Dirichlet forms which are involved in Theorem \ref{Th:SABRDirichletForm}.
For a comprehensive treatment of symmetric and non-symmetric Dirichlet forms, see the monographs \cite{BouleauHirsch,FukushimaOshimaTakeda} and \cite{RoecknerMa} respectively.
\begin{remark}
To extend the above proof to the untruncated problem, see Remark \ref{Rem:ContinuityNonTruncated} above. We remark moreover that the form \eqref{eq:LogvolBilinearformSABRWeighted} is the (unique) Dirichlet form corresponding to the (SABR) Markov process \eqref{eq:SABRPDELogVol}. Recall, that the law of the process \eqref{eq:LogvolBilinearformSABRWeighted} is unique, by pathwise uniqueness of the solutions of \eqref{eq:SABRSDE} when imposing Dirichlet boundary conditions at zero, and the corresponding martingale problem (see \cite[Theorem 21.7]{kallenberg2002foundations}, and \cite[Lemma 21.17]{kallenberg2002foundations}) is well-posed. 
\end{remark}

\section{Discretization}\label{Sec:Discretization}
In this section we derive a suitable discretization in space and time for the variational formulation of the SABR pricing equations.
We propose a multiresolution approximation inspired by the (unweighted) wavelet discretization in \cite[Section 3.4]{PetersdorffSchwab}. To accommodate to the current setting, we shall rely on the weighted multiresolution analysis established in \cite[Sections 5.2 and 5.3]{BeuchlerSchneiderSchwab}, for further reference see also \cite{ReichmannSchwab}.
\subsection{Space discretization and the semidiscrete problem}\label{Sec:SpaceDiscretization}

Given $u_0 \in {\V}$
and $g \in L^2(J;{\V^*})$, we first choose an approximation $u_{(0,L)}\in {\V}_L$ of the initial data $u_0$, where $\V^L \subset \V$ is a finite dimensional subspace. Then the semi-discrete problem reads as follows: Find $u_L \in H^1(J;{\V}_L)$, such that
\begin{align}
&u_L(0)=u_{(0,L)}\label{eq:Initialdata}\\
&(\tfrac{d}{dt}u_L, v_L)_{\Hh}+ a(u_L,v_L)=(g(t),v_L)_{{\V}^*\times {\V}}, \quad \forall v_L \in {\V}_L\label{eq:Semidiscrete}.
\end{align}
We assume $u_{(0,L)}=P_Lu_0$, were $P_L:{\V}\rightarrow {\V}_L$, $u\longmapsto u_L$ is an appropriate projector (see equation \eqref{eq:TruncatedProjecor2dim} below). 
The semi-discrete problem is an initial value problem for $N=\dim \V^L$ ordinary differential equations
\begin{equation}\label{eq:SemidiscreteProb}
\bold{M}\tfrac{d}{dt}\underline{u}(t) + \bold{A} \underline{u}(t)=\underline{g}(t), \quad \underline{u}(0)=u_0,
\end{equation}
where $\underline{u}(t)$ denotes the coefficient vector of $u_L(t)$  
for $t\geq 0$, 
and $\bold{M}$ and $\bold{A}$ denote the mass and stiffness matrix respectively with respect to some basis of the discretization space $\V^L$, which we construct in the subsequent sections.
\subsubsection{Discretization spaces}
Recall that the bivariate pricing equation \eqref{eq:SABRPDELogVol} is parabolic with degenerate elliptic operator $A$. 
To accommodate to the degeneracy of $A$ in \eqref{eq:GeneratorLogscaleSABR}, we introduced the weighted spaces $\Hh$ and $\V$ with weights which are singular at the boundary $x=0$ of the domain $G$ (cf. Definitions \ref{Def:HilbertSpaceSABR} and \ref{Def:WeightedSobolevSpaceVSABR}) to establish well-posedness of the variational formulation of the pricing equations. 
To construct finite element approximation spaces for ${\V}$ and ${\Hh}$ we first establish univariate approximation in each dimension separately. 
The approximation spaces to the (weighted) univariate spaces $V_x,\ V_y$ and $H_x, \ H_y$ are then assembled to obtain the bivariate approximation spaces to $\V$ and $\Hh$. 
From now on we restrict ourselves without loss of generality to the unit interval $I=[0,1]$ if not stated otherwise. 
Our multiresolution analysis on $I\subset \R$ consists of a nested family of spaces
\begin{equation}\label{eq:NestedFamily}
V^0 \subset V^1 \subset \ldots \subset V^L \subset \ldots \subset L^2(I,{\omega})=:H(I,{\omega}), 
\end{equation}
where the choice of the spaces $V^i$ (this choice is specified in Section \ref{Sec:Wavelets} below) is such, that the inclusions \eqref{eq:NestedFamily} are valid, and such that $\overline{\bigcup_{l \in \mathbb{N}} V^l}= L^2(I,{\omega})$ holds, where $L^2(I,{\omega})$ denotes a space of square integrable functions on $I$ with weight function $\omega$. The latter inclusion and convergence statements for our choice of spaces $V^i$ are justified in Theorem \ref{Th:Normequivalences} in the following section.\\
To specify the choice of $V^l$ in \eqref{eq:NestedFamily}, we first specify our choice of partitions of the unit interval $I$ at discretization level $L\in \mathbb{N}$:  For this, let $L\in \mathbb{N}$ denote the discretization level, and let $\mathcal{T}_0$ be a given initial partition of $I$. We assume that for any $L>0$, the family $\{\mathcal{T}_0,\ldots,\mathcal{T}_L\}$ of partitions of the unit interval $I$ is such, that for each $0<l<L$ the partition $\mathcal{T}_l$ is obtained from the partition $\mathcal{T}_{l-1}$ by bisection of each of its subintervals. Hence, for any $l \in \{0,\ldots, L\}$, the partition $\mathcal{T}_l$ has $N^l=C 2^l$ subintervals, where $C$ denotes the number of intervals of the initial partition $\mathcal{T}_0$. 
Note that with this discretization, for any $l\in \{0,\ldots, L\}$ the \emph{dimension} of the space $V^l$ (specified in Section \ref{Sec:Wavelets}) is 
\begin{align}\label{DimensionVl}\begin{split}
&\dim V^l=C \ 2^l=:N^l,\quad 0\leq l < L, \quad \dim V^L=C \ 2^L=:N
\end{split}
\end{align}
for a constant $C>0$.
Furthermore, the codimension on each level $l$ is
\begin{align}\label{CodimensionVl}
M^l:=N^{l+1}-N^{l},\quad 0\leq l < L.
\end{align}
\subsubsection{Wavelets for $L^2$-spaces on an interval}\label{Sec:Wavelets}
As announced above, we now specify the choice of  $V^l$ in \eqref{eq:NestedFamily}. For the univariate approximation spaces of the
$H_x:=L^2(I,\omega_0^x), \ H_y:=L^2(I)$ such as $V_x:=H^1(I,\omega_1^x),\ V_y:=H^1(I)$ on the interval $I$, we recapitulate basic concepts and definitions of (bi-)orthogonal, compactly supported wavelets from \cite[Section 2]{BeuchlerSchneiderSchwab},         and \cite[Section 6.2]{ReichmannDissertationThesis}. 
We consider 
two-parameter wavelet systems $\{ \psi_{l,k} \}$, $l=0,\ldots,\infty$,  $k\in \{1,\ldots M^l\}$ of compactly supported functions $\psi_{l,k}$. Here the first index, $l$, denotes the ``level'' of refinement resp. resolution: wavelet functions $\psi_{l,k}$ with large values of the level index are well-localized in
the sense that $\textrm{diam supp}(\psi_{l,k} ) = \mathcal{O}(2^{-l})$.
The second index, $k \in M^l$, measures the localization of wavelet $\psi_{l,k}$ within the interval $I$ at scale $l$ and ranges in the index set $M^l$. We refer to \cite[Section 4]{ReichmannSchwab} for a graphic illustration of the construction in a specific example.
In order to achieve maximal flexibility in the construction of wavelet systems (which
can be used to satisfy other requirements, such as minimizing their support size or to
minimize the size of constants in norm equivalences, see \eqref{eq:Normequ},\eqref{eq:NormequWeighted} cf. \cite[Section 6.2]{ReichmannDissertationThesis}), we propose wavelet bases, which are biorthogonal in $L^2 (I)$. These consist of a primal wavelet system $\{\psi_{l,k} \}$, $l=0,...,\infty,$ $k\in M^l$ which is a Riesz basis of $L^2 (I)$ (and which enter explicitly in the space discretizations) and a corresponding dual wavelet system $\{\widetilde \psi_{l,k} \}$,  $l=0,...,\infty$, $k\in M^l$, which are not used explicitly in the algorithms, see \cite[Section 6.2]{ReichmannDissertationThesis}.
The primal wavelet bases $\psi_{l,k}$ span the finite dimensional spaces
\begin{equation}\label{eq:Discrspace1dimSpan}\begin{split}
 V^l&=\textrm{span}\{\psi_{i,j} \ | \ 0\leq i\leq l, 1\leq j\leq M^l\},\quad l>0, \quad \textrm{that is}
\quad V^l=
V^{l-1}\bigoplus W^l,
\end{split}\end{equation}
that is, $V^l=V^{l-1}\bigoplus W^l$ 
inductively, where $W^l=\textrm{span}\{ \psi_{l,1},\ldots, \psi_{l,M^l} \}$. Hence, 
\begin{equation}\label{eq:Discrspace1dim}
V^L=\bigoplus_{0\leq l< L}W^l, \ \textrm{ where } \ W^l:=\textrm{span}\{ \psi_{l,1},\ldots, \psi_{l,M^l} \}
\end{equation}
and dual spaces are defined analogously via the dual wavelet system
\begin{equation*}
\widetilde V^{L}:= \bigoplus_{0\leq l< L}\widetilde W^l, \ \textrm{ where } \ \widetilde W^l:=\textrm{span}\{ \widetilde \psi_{l,1},\ldots, \widetilde \psi_{l,M^l} \}.
\end{equation*}
Furthermore, the spaces $\bigoplus_{l=0}^{\infty} W^l$ and $\bigoplus_{l=0}^{\infty}\widetilde W^l$ are assumed to be dense in $L^2(I)$.
To construct such spaces, we consider a multiresolution basis $\{\psi_{k,l}\}_{(k,l)}$ of $L^2(I,\omega)$ 
with the following properties:
\begin{enumerate}
\item The basis functions $\psi$, $\widetilde \psi$ are biorthogonal in $L^2(I)$, that is 
\begin{equation*}
 \int_0^1\psi_{k,l}(x)\widetilde{\psi}_{k',l'} dx=\delta_{k,k'}\delta_{l,l'}.
\end{equation*}
\item\label{WaveletScalingWeighted} The wavelets $\psi_{k,l}$ and their duals $\widetilde \psi_{k,l}$ are local with respect to the corresponding scale and normalized, that is
\begin{equation*}
\textrm{diam supp}(\psi_{k,l})=C_{\psi_{k,l}}2^{-l}\quad \textrm{such as} \quad||\psi_{k,l}||_{L^1}=C_{\psi}' 2^{-l/2}
\end{equation*}
holds, where the constants $C_{\psi}$, $C_{\psi}'$ may depend on the ``mother'' wavelet.
\item The primal wavelets satisfy a vanishing moment condition
\begin{equation*}
\int_0^1\psi_{k,l}(x)\ x^{\alpha}dx=0, \quad \textrm{for} \quad \alpha=0,\ldots,p,
\end{equation*}
where $p$ denotes the polynomial order of the wavelets, see \cite[Equations (2.2),(2.4)]{BeuchlerSchneiderSchwab} such as \cite[Equations (6.7),(6.8)]{ReichmannDissertationThesis}. 
The dual wavelets except the ones at the endpoints satisfy
\begin{equation*}
\int_0^1\widetilde \psi_{k,l}(x)\ x^{\alpha}dx=0, \quad \textrm{for} \quad \alpha=0,\ldots,p+1,
\end{equation*}
while at the end points the dual wavelets satisfy only 
\begin{equation*}
\int_0^1\widetilde \psi_{k,l}(x)\ x^{\alpha}dx=0, \quad \textrm{for} \quad \alpha=1,\ldots,p+1. 
\end{equation*}
This third condition implies that the wavelets satisfy the zero Dirichlet
condition namely $\psi_{k,l}(0) =\psi_{k,l} (1) = 0$, cf. \cite[Equation (6.8)]{ReichmannDissertationThesis}.\end{enumerate}
In the stock price dimension, the weighted spaces $H_x$ and $V_x$ such as their wavelet bases further satisfy (cf. \cite[Assumption 3.1 and 3.2]{BeuchlerSchneiderSchwab})
\begin{enumerate}
\item[(w1)]\label{Assumption3.1} The weight function $\omega(x)$ belongs to $W^{1,\infty}((\delta,1))$ for every $\delta>0$ and satisfies 
\begin{equation*}
C_{\omega}^{-1}\leq \frac{\omega(x)}{x^{\alpha}}\leq C_{\omega}, \qquad C_{\omega}^{-1}\leq \frac{\omega'(x)}{x^{\alpha-1}}\leq C_{\omega}
\end{equation*}
for some $\alpha \in \R$ and a constant $C_\omega>0$, which only depends on the weight function.
\item[(w2)]\label{Assumption3.2} The boundary wavelets $\psi_{k,l}(x)$ 
and dual wavelets $\widetilde \psi_{k,l}(x)$ are denoted by the index set
 \begin{align*}
 &\nabla_l^L=\{k\in \mathbb{N}_0, \gamma-1\leq k \leq 2^{l}-1, 0\in \textrm{supp}\ \psi_{k,l} \},\\
 &\widetilde{\nabla_l}^L=\{ k\in \mathbb{N}_0, \gamma-1\leq k \leq 2^{l}-1, 0\in \textrm{supp}\ \widetilde \psi_{k,l} \},
\end{align*}
and the boundary wavelets and their duals satisfy the conditions
\begin{equation*}
\begin{array}{rl}
&|\psi_{k,l}(x)|\leq C_{\psi}2^{l/2}(2^l x)^{\gamma},\\
&|(\psi_{k,l})'(x)|\leq C_{\psi}2^{3l/2}(2^l x)^{\gamma-1}, \quad \gamma\in \mathbb{N}_0, \ k\in \nabla_l^L\\
&|\tilde \psi_{k,l}(x)|\leq C_{\psi}2^{l/2}(2^l x)^{\tilde \gamma},\\
&|(\tilde \psi_{k,l})'(x)|\leq C_{\psi}2^{3l/2}(2^l x)^{\tilde \gamma-1}, \quad  \tilde \gamma\in \mathbb{N}_0, \ k\in \widetilde \nabla_l^L,
\end{array}
\end{equation*}
where $\gamma, \widetilde \gamma \in \mathbb{N}_0$ are parameters such that $\alpha+\gamma>-\frac{1}{2}$ and $-\alpha+\tilde \gamma>-\frac{1}{2}$ is fulfilled for the parameter $\alpha$ in $(w1)$.
We refer to \cite{DahmenKunothUrban} for explicit constructions.
\end{enumerate}
It follows that any $v\in V$ has a representation as a series and any $v_L \in V^L$ as a linear combination
\begin{equation}\label{eq:ReprElemV}
v_L=\sum_{l=0}^{L}\sum_{j=1}^{M^l}v_j^l \psi_{l,j}, \quad v_L \in V^L; \qquad  v=\sum_{l=0}^{\infty}\sum_{j=1}^{M^l}v_j^l \psi_{l,j},  \quad v\in V,
\end{equation}
where $v_j^l=(v,\widetilde \psi_{l,j})_{L^2(I)}$, 
cf. 
\cite[Equation (2.9)]{BeuchlerSchneiderSchwab}.
Approximations of functions in $V$ are obtained by truncating the wavelet expansion, cf. \cite[Equation (3.13)]{PetersdorffSchwab} or \cite[p. 161]{ReichmannComputationalMethods}.
\begin{definition}[Projection Operator]\label{def:TruncatedProjecor}
For a subspace $V$ of a (possibly weighted) Hilbert space $L^2(I,\omega)$ over an interval $I$ (cf. \eqref{eq:NestedFamily}) the projection to the univariate finite element discretization space $P_L:V\rightarrow V^L$ is defined by truncating the wavelet expansion at refinement level $L$
\begin{equation}\label{eq:TruncatedProjecor}
P_L v:= \sum_{l=0}^{L}\sum_{j=1}^{M^l}v_j^l\psi_{l,j}, \quad v\in V.
\end{equation}
\end{definition}

\subsubsection{Norm equivalences}
Wavelet norm equivalences are akin to the classical Parseval relation in Fourier analysis (see \cite[Section 12.1.2]{ReichmannComputationalMethods}) allowing to express Sobolev norms in terms of sums of its Fourier coefficients.
Norm equivalences are relevant for the construction of the mass- and stiffness matrices and for approximation estimates in the error analysis, cf. \cite[Equations (2.5), and (2.6)]{BeuchlerSchneiderSchwab}.
In the unweighted univariate case  
  there hold the standard norm\footnote{With a view to error analysis, it is conventional (cf. \cite[Section 3.1]{PetersdorffSchwab} and \cite[Section 3.6.1]{ReichmannComputationalMethods}) to consider functions in $V$, which have additional regularity:
In the unweighted univariate case one considers the classical Sobolev spaces $H^t_0(I)$, $t=0,1,2$, where the subscript denotes Dirichlet boundary conditions.} equivalences 
\begin{equation}\label{eq:Normequ}
||u||_{H^{s}_0(I)}^2\approx\sum_{l=0}^{\infty}\sum_{k=0}^{2^l}2^{2ls}|u_k^l|^2.
\end{equation} 

For the stock price dimension $V_x$ we need norm equivalences 
in weighted spaces, for which the requirements (\eqref{Assumption3.1} and \eqref{Assumption3.2}, cf. \cite[Section 3]{BeuchlerSchneiderSchwab}) are posed on the wavelets and their duals.
\begin{theorem}[Weighted norm equivalences]\label{Th:Normequivalences}
Under the assumptions of Section \ref{Sec:Wavelets} on the wavelet basis of the discretization spaces, there holds for any $u \in L^2(I, \omega)$ the norm equivalence
of its $L^2(I, \omega)$-norm and of the discrete $l^2_{\omega}$-norm of its coefficients with respect to the wavelet basis.
\begin{align}\label{eq:NormequWeighted}
||u||_{L^2(I,\omega)}^2&\approx\sum_{l=0}^{\infty}\sum_{k=0}^{M^l}\omega^2(2^{-l}k)|u_k^l|^2
\end{align}
\end{theorem}
\begin{proof}
See \cite[Theorem 3.3]{BeuchlerSchneiderSchwab}, and also \cite[Theorem 5.1]{BeuchlerSchneiderSchwab}.
\end{proof}
\subsubsection{Bivariate setting}
For the bivariate case we set, similarly as above, without loss of generality $G=(0,1)\times(0,1)$.  
The Hilbert spaces
$\Hh^t(G, \omega)$, $t=0,1,2$ can be constructed from $H^t_x(I,\omega_t^x)$ and $H^t_y(I)$ via tensor products, see Appendix \ref{Sec:TensorHilbertSpaces} for explicit constructions.
We define the two-dimensional discretization spaces as the tensor product of the univariate discretization spaces for the $x$ and $y$ coordinates
\begin{equation}\label{eq:Def2dimdiscrspaceTensor}
\mathcal{V}_L:=V_{L_x}\otimes V_{L_y}. 
\end{equation}
Therefore, it holds similarly to \eqref{eq:Discrspace1dimSpan} that
\begin{equation}\label{eq:Discrspace2dimSpan}
\mathcal{V}_L=\textrm{span}\left\{ \psi_{\bold{l,k}}:0\leq l_x,l_y \leq L,\ 1\leq k_x\leq M^l_x, 1\leq k_y\leq M^l_y\right\},
\end{equation}
where $\psi_{\bold{l,k}}(x,y)=\psi_{(l_x,l_y),(k_x,k_y)}:=\psi_{l_x,k_x}(x)\psi_{l_y,k_y}(y)$, for $(x,y)\in(0,1)\times(0,1)$. Recall from \eqref{eq:Discrspace1dimSpan} that $W^{l_x}=\textrm{span}\{ \psi_{l_x,1},\ldots, \psi_{l_x,M^{l_x}} \}$ and  $W^{l_y}=\textrm{span}\{ \psi_{l_y,1},\ldots, \psi_{l_y,M^{l_y}} \}$ are the corresponding complement spaces.  Hence it follows from \eqref{eq:Discrspace1dim} with \eqref{eq:Def2dimdiscrspaceTensor} directly that
 in the bivariate case
\begin{equation}\label{eq:Discrspace2dim}
 \mathcal{V}_L=\bigg(\bigoplus_{0\leq l_x\leq L}W^{l_x}\bigg)\otimes\bigg(\bigoplus_{0\leq l_y\leq L}W^{l_y}\bigg)
=\bigoplus_{0\leq l_x, l_y\leq L} W^{l_x}\otimes W^{l_y}.
\end{equation}
Every element $u \in \Ll^2(G)$ has\footnote{Note that $\Hh=\Ll^2(G,x^{\mu/2})\subset \Ll^2(G)$.} a series representation cf. \cite[Equation (13.6)]{ReichmannComputationalMethods}
\begin{equation}\label{eq:ReprElemV2dim}
u=\sum_{l_x,l_y=0}^{\infty}\sum_{k_x=1}^{M^{l_x}}\sum_{k_y=1}^{M^{l_y}}
u_{\bold{k}}^{\bold{l}} \psi_{\bold{l,k}},  \qquad \textrm{where} \ \bold{l,k}=(l_x,l_y),(k_x,k_y),
\end{equation}
and the projection operator $P_L:\V\rightarrow \mathcal{V}_L$, $u\mapsto P_Lu=:u_L$ is defined as in \eqref{eq:TruncatedProjecor} by truncating  at level $L$ the above series representation,
\begin{equation}\label{eq:TruncatedProjecor2dim}
u_L=\sum_{l_x,l_y=0}^{L}\sum_{k_x=1}^{M^{l_x}}\sum_{k_y=1}^{M^{l_y}}
u_{\bold{k}}^{\bold{l}} \psi_{\bold{l,k}},  \quad  u\in V, \quad \bold{l,k}=(l_x,l_y),(k_x,k_y).
\end{equation}
It is immediate from the underlying tensor product structure and from the norm-equivalences \eqref{eq:Normequ} of the one-dimensional case, that in the bivariate case there holds the norm equivalence
\begin{equation}\label{eq:Normequ2dim}
||u||^2_{\Hh^{\bold{s}}(G)}\approx \sum_{l_x,l_y=0}^{\infty} \sum_{k_x=1}^{2^{l_x}}\sum_{k_x=1}^{2^{l_x}} \left(2^{2s_xl_x}+2^{2s_yl_y}\right)|u_{\bold{k}}^{\bold{l}}|^2 ,
\end{equation}
for the Sobolev spaces $H^k_0(G)$, $k=0,1,2$, where $\bold{l,k}=(l_x,l_y),(k_x,k_y),$ as in \eqref{eq:ReprElemV}, and $\bold{s}=(s_x,s_y)$ 
cf. \cite[Equation 13.8]{ReichmannComputationalMethods}. For notational simplicity we stated here the unweighted version of the bivariate norm equivalence.
Note however, that passing from the univariate to the bivariate case is a direct consequence of the tensor product structure (see Section \ref{Sec:TensorHilbertSpaces}) and is valid both in unweighted and weighted Sobolev spaces 
analogously.\\

Now we are in a position to calculate the matrices $\bold{M}^x, \bold{B}^x$, and $\bold{S}^x$ such as $\bold{M}^y, \bold{B}^y$, and $\bold{S}^y$ as building blocks of the mass- and stiffness matrices $\bold{M}$ and $\bold{A}$ appearing in the semi-discrete problem \eqref{eq:SemidiscreteProb} with respect to the constructed wavelet basis $\{\psi_{\bold{l},\bold{k}}\}$: for the $y$-coordinate, 
the matrices in the appropriate weighted $L_{\omega}^2$-norm read
\begin{equation}\label{eq:Matricesy}
\begin{array}{ll}
\bold{M}_{\omega_{y,M}^2}^y&:=\left(\int_0^1\frac{\omega_y^2(y)\psi_{l_y,k_y}(y)\psi_{l_y',k_y'}(y) }{\omega(2^{-l_y}k_y)\omega(2^{-l_y'}k_y')}dy\right)_{0\leq l_y',l_y\leq L; \ 0 \leq k_y' \leq 2^{l_y'},\ 0\leq k_y \leq 2^{l_y}}\\
\bold{S}_{\omega_{y,S}^2}^y&:=\left(\int_0^1\frac{\omega_y^2 (y)\psi_{l_y,k_y}'(y)\psi_{l_y',k_y'}'(y) }{\omega(2^{-l_y}k_y)\omega(2^{-l_y'}k_y')}dy\right)_{0\leq l_y',l_y\leq L; \ 0 \leq k_y' \leq 2^{l_y'},\ 0\leq k_y \leq 2^{l_y}}\\
\bold{B}_{\omega_{y,B}^2}^y&:=\left(\int_0^1\frac{\omega_y^2 (y)\psi_{l_y,k_y}'(y)\psi_{l_y',k_y'}(y) }{\omega(2^{-l_y}k_y)\omega(2^{-l_y'}k_y')}dy\right)_{0\leq l_y',l_y\leq L; \ 0 \leq k_y' \leq 2^{l_y'},\ 0\leq k_y \leq 2^{l_y}},
\end{array} 
\end{equation}
and for the $x$-coordinate, the corresponding matrices are
\begin{equation}\label{eq:Matricesx}
\begin{array}{ll}
\bold{M}_{\omega_{x,M}^2}^x&:=\left({\int_0^1\frac{\omega_x^2 (x)\psi_{l_x,k_x}(x)\psi_{l_x',k_x'}(x) }{\omega(2^{-l_x}k_x)\omega(2^{-l_x'}k_x')}dx}\right)_{0\leq l_x',l_x\leq L; \ 0 \leq k_x' \leq 2^{l_x'},\ 0\leq k_x \leq 2^{l_x}}\\
\bold{S}_{\omega_{x,S}^2}^x&:=\left({\int_0^1\frac{\omega_x^2 (x)\psi_{l_x,k_x}'(x)\psi_{l_x',k_x'}'(x) }{\omega(2^{-l_x}k_x)\omega(2^{-l_x'}k_x')}dx}\right)_{0\leq l_x',l_x\leq L; \ 0 \leq k_x' \leq 2^{l_x'},\ 0\leq k_x \leq 2^{l_x}}\\
\bold{B}_{\omega_{x,B}^2}^x&:=\left({\int_0^1\frac{\omega_x^2 (x)\psi_{l_x,k_x}'(x)\psi_{l_x',k_x'}(x) }{\omega(2^{-l_x}k_x)\omega(2^{-l_x'}k_x')}dx}\right)_{0\leq l_x',l_x\leq L; \ 0 \leq k_x' \leq 2^{l_x'},\ 0\leq k_x \leq 2^{l_x}},
\end{array} 
\end{equation}
where $\omega_x,$ and $\omega_y$ denote weight functions in the respective dimensions. Our equations \eqref{eq:Matricesy} and \eqref{eq:Matricesx} closely follow the construction of \cite[equation (3.12)]{BeuchlerSchneiderSchwab} for the univariate weighted matrices. We take these as building blocks for the construction of our bivariate matrices, displayed in \eqref{eq:MassMatrix} and \eqref{eq:StiffnessMatrix}, with the appropriate choice of weight functions ($\omega_{x,M}^2=x^{\mu}$, $\omega_{y,M}^2=1$ in \eqref{eq:MassMatrix}, and further $\omega_{y,M}^2=e^{2y}, \omega_{y,S}^2=1, \omega_{y,B}^2=e^{y}$ such as $\omega_{x,M}^2=x^{\mu}, \omega_{x,S}^2=x^{2\beta+\mu}, \omega_{x,B}^2=x^{\beta+\mu}$ in \eqref{eq:StiffnessMatrix}).
Similar constructions in the standard case can be found in \cite{ReichmannComputationalMethods}. In our case, the stiffness matrix $\bold{A}$ in the semi-discrete problem \eqref{eq:SemidiscreteProb} is
\begin{equation}
\begin{array}{ll}
\bold{A}&=\left(\bold{A}_{(\bold{l',k'}),(\bold{l,k})}\right)_{0\leq l_x',l_x\leq L; \ 0 \leq k_x' \leq 2^{l_x'},\ 0\leq k_x \leq 2^{l_x}}\\
&:=\left(a(\psi_{\bold{l,k}},\psi_{\bold{l',k'}})\right)_{0\leq l_x',l_x\leq L; \ 0 \leq k_x' \leq 2^{l_x'},\ 0\leq k_x \leq 2^{l_x}},
\end{array}
\end{equation}
where $a(\cdot,\cdot)$ denotes the SABR-bilinear form \eqref{eq:LogvolBilinearformSABRWeighted}.
With respect to the weighted multiresolution basis $\{\psi_{\bold{l,k}}\psi_{\bold{l',k'}}\}$ defined in this section, the mass matrix reads
\begin{equation}\label{eq:MassMatrix}
\bold{M}=\bold{M}^x_{x^{\mu}}\otimes \bold{M}^y_{1}, \phantom{\left(c_{x_1} \bold{B}_{x^{2\beta+\mu-1}}^x \otimes \bold{M}_{e^{2y}}^y + c_{x_2} \bold{B}_{x^{\beta+\mu}}^x \otimes  +c_y \bold{M}_{x^{\mu}}^x \right)}
\end{equation}
and the stiffness matrix $\bold{A}$ takes the form
\begin{align}\label{eq:StiffnessMatrix}\begin{split}
\bold{A}=&
\left(\mathcal{Q}_{xx} \bold{S}_{x^{2\beta+\mu}}^x \otimes \bold{M}_{e^{2y}}^y +\mathcal{Q}_{xy} \bold{B}_{x^{\beta+\mu}}^x \otimes \bold{B}_{e^y}^y   + \mathcal{Q}_{yy} \bold{M}_{x^\mu}^x \otimes \bold{S}^y_1\right)\\
&+\left(c_{x_1} \bold{B}_{x^{2\beta+\mu-1}}^x \otimes \bold{M}_{e^{2y}}^y + c_{x_2} \bold{B}_{x^{\beta+\mu}}^x \otimes \bold{M}_{e^{y}}^y +c_y \bold{M}_{x^{\mu}}^x \otimes \bold{B}^y_1\right),
\end{split}\end{align}
where the coefficients are $(\mathcal{Q}_{xx},\mathcal{Q}_{xy},\mathcal{Q}_{yy})=(\tfrac{1}{2}, \rho \nu,\tfrac{\nu^2}{2})$ and $(c_{x_1},c_{x_2},c_{y})=(\tfrac{2\beta+\mu}{2},\rho\nu,\tfrac{\nu^2}{2})$.
\subsection{Time discretization and the fully discrete scheme}\label{Sec:TimeDiscretization}
In this section we define a $\theta$-scheme for our time discretization to introduce the fully discrete scheme of our finite element method. Furthermore, following \cite{PetersdorffSchwab} we conclude (see Proposition \ref{Prop:StabilityThetaScheme} below) that the stability of the $\theta$-scheme remains valid in our setting.
For this, we introduce the following dual norm for the approximation spaces:
\begin{equation}\label{eq:DualNorm}
||f||_*:=\sup_{v_L\in V^L}\frac{(f,v_L)_{\V^*\times \V}}{||v_L||_{\V}},\qquad v_L \neq 0,\qquad f\in (V^L)^*,
\end{equation}
furthermore, for $T<\infty$ and $M\in \mathbb{N}$ we consider the following uniform time-step and time mesh
\begin{equation}\label{eq:TimeStep}
 k:=\tfrac{T}{M}, \qquad \textrm{and}\qquad t^m=mk,\ m=0,\ldots,M.
\end{equation}
The $\theta$-scheme for the time-discretization and the fully discrete scheme are described as follows:
\begin{definition}[$\theta$-scheme and the fully discrete scheme]\label{def:ThetaScheme}
Given the initial data $u_L^0:=u_{(0,L)}=P_Lu^0$, for the projector in \eqref{eq:TruncatedProjecor2dim} for $m=0,\ldots, M-1$ find $u_L^{m+1} \in \V^L$ such that for all $v_L \in \V^L$:
\begin{equation}\label{eq:ThetaScheme}
\tfrac{1}{k}(u_L^{m+1}-u_L^{m},v_L)_{\V^*\times \V}+a(\theta u_L^{m+1}+(1-\theta)u_L^{m},v_L)=(\theta g(t^{m+1})+(1-\theta)g(t^m),v_L)_{\V^*\times \V} 
\end{equation}
Hence, the fully discrete finite element scheme for the SABR model reads
\begin{equation}
(\tfrac{1}{k}\bold{M} + \theta \bold{A}) \underline{u}^{m+1}=\tfrac{1}{k}\bold{M}\underline{u}^{m}-(1-\theta)\bold{A}\underline{u}^{m}+\underline{g}^{m+\theta}, \quad m=0,1,\ldots,M-1,
\end{equation}
where $\bold{M}$ denotes the mass matrix \eqref{eq:MassMatrix}, $\bold{A}$ the stiffness matrix \eqref{eq:StiffnessMatrix}, and $\underline{u}^{m}$ is the coefficient vector of $u^m_L$ with respect to the basis of $\V^L$.
\end{definition}
\begin{proposition}[Stability of the $\theta$-scheme]\label{Prop:StabilityThetaScheme}
For $\frac{1}{2}\leq \theta \leq 1$ let the constants $C_1$ and $C_2$ satisfy
\begin{equation}\label{eq:ConstantsC1C2Stability1}
 0<C_1<2, \qquad C_2\geq\frac{1}{2-C_1},
\end{equation}
and for $0\leq \theta < \frac{1}{2}$ denote $\lambda_A=\textrm{sup}_{v_L \in V^L}\frac{||v_L||^2_{\Hh}}{||v_L||^2_*}$, and the constants $C_1$ and $C_2$ be such that
\begin{equation}\label{eq:ConstantsC1C2Stability2}\begin{array}{lll}
0<C_1<2-\sigma, &\qquad C_2\geq \frac{1+(4-C_1)\sigma}{2-\sigma-C_1} &\textrm{where}\qquad \sigma:=k(1-2\theta)\lambda_A<2.
\end{array}
\end{equation}
Then the sequence $\{u_L^m\}_{m=0}^{M}$ of solutions of the $\theta$-scheme \ref{eq:ThetaScheme}
satisfy the stability estimate
\begin{equation}\label{eq:StabilityThetaScheme}
||u_L^M||^2_{\Hh}+C_1k\sum_{m=0}^{M-1}||u_k^{m+\theta}||^2_{\Hh}\leq||u_L^0||^2_{\Hh}+C_2k\sum_{m=0}^{M-1}||g^{m+\theta}||_*^2. 
\end{equation}
\end{proposition}
\begin{proof}
The proof of this proposition is analogous to the one given in \cite[Proposition 4.1]{PetersdorffSchwab}.
\end{proof}

\section{Error estimates}\label{Sec:ErrorEsimates}
Let ${\V}_L$, the finite dimensional approximation space of the solution space $\V$, be as in Section \ref{Sec:Discretization}.
Furthermore, consider $u^m(x):=u(t^m,x)$, with $t^m$, $m=0,\dots,M$ as in \eqref{eq:TimeStep} and $u_L^m$ as in the fully discrete scheme \eqref{eq:ThetaScheme}.
In this section we estimate the error
\begin{equation}\label{eq:ErrorSplitting}
e^m_L:=u^m(x)-u_L^m(x)=(u^m-P_Lu^m)+(P_Lu^m-u_L^m)=:\eta^m+\xi_L^m,  
\end{equation}
for the the time-points $t^m$, $m=0,\ldots, M$, where $P_L:{\V}\rightarrow \mathcal{V}_L$ denotes the projection \eqref{eq:TruncatedProjecor} on the finite element space via truncated wavelet expansion.
In Section \ref{Sec:ApproximationEstimates} we derive estimates for the error $\eta^m$, $m=0,\ldots, M$ (approximation estimates). Section \ref{Sec:DiscretizationError} is devoted to concluding from the results of Section \ref{Sec:ApproximationEstimates} corresponding error estimates for $\xi_L^m$ , $m=0,\ldots, M$ and the convergence of the fully discrete scheme.

\subsection{Approximation estimates}\label{Sec:ApproximationEstimates}
The crucial ingredient of our error analysis is the derivation of approximation estimates which measure the error $\eta^m=(u^m-P_Lu^m)$ between the true solution of the pricing equation at times $t^m$, $m=0,\dots,M$ and its projection to the discretization space in a suitably chosen norm. 
In the unweighted case, such approximation estimates---as in \eqref{ApproximationProperty} below---in the usual $H^k(I)$ (resp. $\Hh^{\bold{k}}(G)$) norm, $k=0,1,2$ are standard, see \cite[Jackson-type estimates p.163]{ReichmannComputationalMethods}.
With view to the ensuing error analysis we derive here (see Section \ref{Sec:ApproximationEstimatesWeighted} below) analogous estimates for weighted Sobolev norms which are suitable to the Gelfand triple $\V\subset H \subset \V^*$ constructed in Section \ref{Sec:WellPosednessSABR} and the corresponding discretization spaces in Section \ref{Sec:SpaceDiscretization}.
\subsubsection{Approximation estimates in the unweighted case}\label{Sec:ApproximationEstimatesUnweighted}
In the unweighted univariate case it is well-known that there exists for all $u \in H^{l}(I)$ with $l =0,1,2$ an element $u_L$ in the corresponding discretization space 
$V^L$, such that $u_L=P_Lu$ and for $k=0,1$ and $l=0,1,2$, $l\geq k$ 
it holds that
\begin{equation}\label{ApproximationProperty}
||u-P_Lu||_{H^k(I)}\leq C 2^{-(l-k)L}||u||_{H^{l}(I)}, 
\end{equation}
where $P_L$ denotes the projector \eqref{eq:TruncatedProjecor}. The existence of such an element is provided by the norm-equivalences \eqref{eq:Normequ}.
The same estimates hold in the two-dimensional case $G=I\times I$, see \cite[Theorem 13.1.2.]{ReichmannComputationalMethods}, 
in particular the approximation rate $(2^{-L})^{(l-k)}$ only depends on the discretization level $2^{-L}$ and is independent of the dimension of the domain $G$. Analogously to the univariate case, for 
$k_x=0,1$ and $l_x=0,1,2$, $l_x\geq k_x$ such as for $k_y=0,1$ and $l_y=0,1,2$, $l_y\geq k_y$, it holds that
\begin{equation}\label{ApproximationProperty2dim}
||u-P_Lu||_{ \Hh^{\bold{k}}(G)}\leq 
\begin{cases}
C 2^{-(\bold{l}-\bold{k})_*L}||u||_{\mathcal{H}^{\bold{l}}(G)}, \qquad \textrm{if}\ \bold{k}\neq0 \ \textrm{or}\ l_x,l_y\neq2,\\
C 2^{-(\bold{l}-\bold{k})_*L} L^{1/2}||u||_{\mathcal{H}^{\bold{l}}(G)}\ \ \textrm{else}, 
\end{cases}
\end{equation}
where we denote $(\bold{l}-\bold{k})_*:=\min\{l_x-k_x,l_y-k_y\}$ and where $P_L$ is the projection operator \eqref{eq:TruncatedProjecor2dim}.
The estimate \eqref{ApproximationProperty2dim} is a direct consequence of \eqref{ApproximationProperty} and the tensor product construction \eqref{eq:Def2dimdiscrspaceTensor}. Analogous
arguments obtain as above for bivariate norm-equivalences \eqref{eq:Normequ2dim}, cf. \cite[Chapter 13]{ReichmannComputationalMethods}. 

\subsubsection{Approximation estimates in the weighted case}\label{Sec:ApproximationEstimatesWeighted}
In the weighted setting, the order of approximation may depend on the norms of the weighted Sobolev spaces, 
in which we measure the error. 
In accord with usual conventions\footnote{Cf. \cite[Section 3.1]{PetersdorffSchwab} and \cite[Section 3.6.1]{ReichmannComputationalMethods} and see also the unweighted case above.} we consider functions in $\V$ with additional regularity
for our error analysis in the weighted setting. For this purpose we consider weighted Sobolev spaces up to second order, $\mathcal{H}^k(G, \omega)$, $k=0,1,2$, where the weight $\omega$ is yet to be chosen suitably to our setting.
We aim to establish approximation estimates of the following form: for any $u \in \mathcal{H}^k(G, \omega)$, $k=0,1,2$ there exists a constant $C>0$ such that an estimate of the following type holds
\begin{equation}\label{eq:ApproximationPropertyWeighted}
||u-P_Lu||_{\mathcal{H}^k(G, \omega)}\leq 
\begin{cases}
C 2^{-c_{\omega}(\bold{l}-\bold{k})_*L}||u||_{\mathcal{H}^{l}(G, \omega)},\qquad \textrm{for}\ \bold{k}\neq0 \ \textrm{or}\ l_x,l_y\neq2,\\
C 2^{-c_{\omega}(\bold{l}-\bold{k})_*L} L^{1/2}||u||_{\mathcal{H}^{\bold{l}}(G, \omega)}\ \ \textrm{else}, 
\end{cases}
\end{equation}
for a constant $c_{\omega}\in \R_+$, which may depend on the choice of the weights in $\mathcal{H}^k(G, \omega)$, $k=0,1,2$, where $(\bold{l}-\bold{k})_*=\min\{l_x-k_x,l_y-k_y\}$ as in \eqref{ApproximationProperty2dim}.
In the following, we will first prove the one-dimensional version of the approximation property in weighted spaces.
More specifically, we show statements analogous to \eqref{ApproximationProperty} on the weighted Sobolev spaces 
$H^k(I,x^{\mu/2})$, $k=0,1,2$ in the $x$ coordinate. For this, we pass for a function $u(t,x,y) \in L^2(J,\V), t \in J, \ (x,y)\in G$ to $u_y(t,x) \in L^2(J,V), t \in J, \ x\in I$, for $y\in I$ (see Appendix \ref{Sec:TensorHilbertSpaces})). Note that for the $y$ coordinate, the unweighted setting prevails and the estimate \eqref{ApproximationProperty} is valid. We then proceed to the bivariate case $G=I\times I$ by constructing tensor products of univariate multiresolution finite element spaces\footnote{See Appendix \ref{Sec:TensorHilbertSpacesExplicit} for an explicit construction.}. Analogously as in the unweighted case (cf. equation \eqref{ApproximationProperty2dim}), the minimum of the obtained one-dimensional estimate then yields the estimate of \eqref{eq:ApproximationPropertyWeighted}, for the  bivariate case $\mathcal{H}^k(G,\omega)$, $k=0,1,2$, (see Section \ref{Sec:ApproximationEstimatesUnweighted} above and \cite[Section 13.1]{ReichmannComputationalMethods}).

\begin{definition}\label{Def:SpacesForApproxProp}
Consider an interval  $I=(0,R)$, $R>0$ and the weighted Sobolev spaces
\begin{equation}\label{eq:DefSpacesForApproxProp}
H^{k}_j(I,x^{\mu/2}):=\{u:I\rightarrow \R \ \textrm{measurable}: \  D^{a} u \in L^2(I, x^{\mu/2+a \beta j}),  \ a \leq k \},\quad k=0,1,2,
\end{equation}
for $j=0,1$, 
with the norm
\begin{align}\label{eq:DefNormSpacesForApproxProp}
&||u||_{H^k_j(I, x^{\mu/2})}^2
:=\sum_{a\leq k}\int_{I} |D^{a}u(x)|^2 x^{\mu+a \beta j} \ dx, \ k=0,1,2.
\end{align}
To ease notation we shall henceforth denote $H^{k}_{j=0}$ by $H^{k}$ and $H^{k}_{j=1}$ by $H^{k}_{1}$, $k=0,1,2$.
\end{definition}

\begin{remark}\label{Rem:CoincidingSpacesEstimates}
Note that for the spaces $H$ and $V$ in Remark \ref{Rem:CEVTriplet} and for the spaces $H^k_j$, $k=0,1,2$, $j=0,1$ in \eqref{eq:DefSpacesForApproxProp} it holds that $H=H^{0}((0,R),x^{\mu/2})=H^{0}_{1}((0,R),x^{\mu/2})$ and the weighted space $V$ satisfies $V=H^1_{1}((0,R), x^{\mu/2})$ and $V \supset H^1((0,R), x^{\mu/2})$ such as the estimate
\begin{equation}\label{eq:EstimateVvsWeightedH1A}
||v||_{V}^2\leq C_R ||v||^2_{H^1((0,R), x^{\mu/2})}, \quad v\in V 
\end{equation}
for any finite $R>0$ and a positive constant $C_R>0$.
\end{remark}

\begin{proposition}\label{Prop:ApproximationPropertyInteger}
The projection operator in \eqref{eq:TruncatedProjecor} satisfies for $k=0,1$ and $l=0,1,2$, $l\geq k$ the estimate
\begin{equation}\label{ApproximationPropertyInteger}
||u-P_Lu||_{H^k(I,x^{\mu/2})}\leq C 2^{-(l-k)L}||u||_{H^{l}(I,x^{\mu/2})}, \qquad u\in H^2(I,x^{\mu/2}).
\end{equation}
\end{proposition}
It is immediate from Remark \ref{Rem:CoincidingSpacesEstimates}, that the approximation estimate \eqref{ApproximationPropertyInteger} readily applies to the CEV model. 
Furthermore, combining estimate for the $x$ coordinate with the unweighted estimate \eqref{ApproximationProperty} for the $y$ coordinate and taking tensor products (see Appendix \ref{Sec:TensorHilbertSpacesExplicit} for an explicit construction of the bivariate spaces)  
yields that the approximation estimate \eqref{ApproximationProperty2dim} remains valid in the (weighted) bivariate case. In contrast to this, measuring the error in the norms $||u||_{H^k_{j=1}(I, x^{\mu/2})}$, $k=0,1,2$, the approximation in the $x$ coordinate dominates the $y$ coordinate, see Remark \ref{Prop:ApproximationPropertyIntegerWeights}.
\begin{remark}\label{Prop:ApproximationPropertyIntegerWeights}
If we consider $j=1$ in Definition \ref{Def:SpacesForApproxProp}, we do not assume any additional integrability requirements on our solution up to first order derivatives. 
In this case we obtain (weaker) approximation estimates where the order of approximation depends on the parameter $\beta$: the projection operator $P_L:V\rightarrow V^L$ in \eqref{eq:TruncatedProjecor} satisfies for $k=0,1$ and $l=0,1,2$, $l\geq k$ the estimate
\begin{equation}\label{ApproximationPropertyIntegerBetter}
||u-P_Lu||_{H^k_{1}(I,x^{\mu/2})}\leq C 2^{-(1-\beta)(l-k)L}||u||_{H^{l}_{1}(I,x^{\mu/2})}, \qquad u\in H^2_1(I,x^{\mu/2}).
\end{equation}
A proof of this remark is delegated to the Appendix \ref{Sec:ProofsApproxAlternative}.
\end{remark}

\begin{proof}[Proof of Proposition \ref{Prop:ApproximationPropertyInteger}]
Let $u\in \V$ and consider $P_L u \in \V^L$.
Then it is immediate from \eqref{eq:ReprElemV} and \eqref{eq:TruncatedProjecor}
that
\begin{equation*}
 u-P_L(u)=\sum_{l=L+1}^{\infty}\sum_{j=1}^{2^l}u_j^l \psi_{l,j}.
\end{equation*}
It directly follows from the norm equivalence \eqref{eq:Normequ} and the scaling of the wavelet basis-elements to unit norm
in $L^2(I)$ that the derivatives satisfy
\begin{align}\label{eq:CalcProofApproxprop1}\begin{split}
||(u-P_Lu)'||_{L^2(I,x^{\mu/2})}^2
&=\tilde C\sum_{l=L+1}^{\infty}2^{2l}\sum_{k=0}^{2^l}\omega^2(2^{-l}k)|u_k^l|^2\\
&\geq \tilde C 2^{2L}\sum_{l=0}^{\infty}\sum_{k=0
}^{2^l}\omega^2(2^{-l}k)|u_k^l|^2
= \tilde C 2^{2L}||u-P_Lu||_{L^2(I,x^{\mu/2})}^2. 
\end{split}
\end{align}
Now with $C=\frac{1}{\tilde C}$ and recalling that we have set $h=2^{-2L}$, the following relation 
is obtained:
\begin{equation}\label{eq:CalcProofApproxprop2}\begin{array}{rl}
||u-P_Lu||_{L^2(I,x^{\mu/2})}^2&\leq C h||(u-P_Lu)'||_{L^2(I,x^{\mu/2})}^2\leq C h ||u'||_{L^2(I,x^{\mu/2})}^2\\
& \leq C h \left(||u||_{L^2(I,x^{\mu/2})}^2+||u'||_{L^2(I,x^{\mu/2})}^2\right)= C h ||u||_{H^1(I,x^{\mu/2})}^2.
\end{array}
\end{equation}
Analogously, replacing $(u-P_Lu)'$ by $(u-P_Lu)''$ and $(u-P_Lu)$ by $(u-P_Lu)'$ in equations \eqref{eq:CalcProofApproxprop1} and 
\eqref{eq:CalcProofApproxprop2}, one obtains
\begin{equation}\label{eq:CalcProofApproxprop3}\begin{array}{rll}
||(u-P_Lu)'||_{L^2(I,x^{\mu/2})}^2&\leq C h||(u-P_Lu)''||_{L^2(I,x^{\mu/2})}^2\leq C h ||u''||_{L^2(I,x^{\mu/2})}^2,
\end{array}
\end{equation}
and adding up \eqref{eq:CalcProofApproxprop2} and \eqref{eq:CalcProofApproxprop3} yields:
\begin{equation}\label{eq:CalcProofApproxprop4}\begin{array}{rll}
||u-P_Lu||_{H^1(I,x^{\mu/2})}^2&=||u-P_Lu||_{L^2(I,x^{\mu/2})}^2+ ||(u-P_Lu)'||_{L^2(I,x^{\mu/2})}^2\\
& \leq C h \sum_{k=0}^2||u^{(k)}||_{L^2(I,x^{\mu/2})}^2
=C h ||u||_{H^2(I,x^{\mu/2})}^2.
\end{array}
\end{equation}
Finally, concatenating \eqref{eq:CalcProofApproxprop1} and 
\eqref{eq:CalcProofApproxprop2} directly yields:
\begin{equation}\label{eq:CalcProofApproxprop5}\begin{array}{rl}
||u-P_Lu||_{L^2(I,x^{\mu/2})}^2&\leq C h||(u-P_Lu)'||_{L^2(I,x^{\mu/2})}^2\leq C (2^{-L})^2||(u-P_Lu)''||_{L^2(I,x^{\mu/2})}^2\\
& \leq C (2^{-L})^2 ||u''||_{L^2(I,x^{\mu/2})}^2\leq C (2^{-L})^2 ||u||_{H^1(I,x^{\mu/2})}^2.
\end{array}
\end{equation}
\end{proof}

\subsection{Discretization error and convergence of the finite element method}\label{Sec:DiscretizationError}
In this section we apply the estimates of Section \ref{Sec:ApproximationEstimates} to derive estimates on the discretization error (cf. equation \eqref{eq:ErrorSplitting}) and to conclude
the convergence of the proposed finite element approximation of the variational solution of the SABR pricing equations.
For the estimates of the discretization error we follow the (unweighted) analysis of \cite[Section 5]{PetersdorffSchwab}.
Corresponding proofs prevail with minor modifications and are provided---accommodated to our setting and notations---in the Appendix \ref{Sec:Proofs} for easy reference.
\begin{lemma}\label{Lem:ErrorThetaScheme}
For $u\in C^1(\bar{J};\mathcal{H}^2(G,\bold{a}))$, the errors $\xi_L^m$ are the solutions of the $\theta$-scheme:\\
Given $\xi_L^0:=P_Lu^0-u_L^0$, for $m=0,\ldots, M-1$ find $\xi_L^{m+1} \in V^L$ such that for all $v^L \in V^L$:
\begin{align}\label{eq:ErrorThetaScheme}\begin{split}
\tfrac{1}{k}&(\xi_L^{m+1}-\xi_L^{m},v_L)_{\V \times \V^*}+a(\theta \xi_L^{m+1}+(1-\theta)\xi_L^{m},v_L)=:(r^m,v_L)_{\V \times \V^*}
\end{split}\end{align}
\end{lemma}
The proof of the above Lemma relies on the observation that the errors $\xi$ satisfy the same PDE as the solutions $u$, therefore the $\theta$ schemes for $\xi$ and for $u$ are analogous\footnote{See Appendix \ref{Sec:Proofs} for details, where we also included for completeness a reminder of the proof of \cite[Lemma 5.1]{PetersdorffSchwab}---accommodated to our notation---which directly carries over to the present situation.}.
The following corollary is a direct consequence of Lemma \ref{Lem:ErrorThetaScheme} and of the the stability of the $\theta$-scheme, established in Proposition \ref{Prop:StabilityThetaScheme}.
\begin{corollary}\label{Cor:StabilityErrorThetaScheme}
There exist constants $C_1$ and $C_2$ independent from the discretization level $L$ and time mesh $k$ such that the solutions of \eqref{eq:ErrorThetaScheme} satisfy the estimate
\begin{equation}\label{eq:StabilityErrorThetaScheme}
||\xi_L^M||^2_{\Hh}+ C_1 \ k \sum_{m=0}^{M-1}||\xi_L^{m+\theta}||_a^2 
\leq  ||\xi_N^0||^2_{\Hh}+C_2 \ k \sum_{m=0}^{M-1}||r^{m}||_{*}^2,
\end{equation}
where for any $f\in (V^L)^*$,
$||f||_*:=\sup_{v_L\in \mathcal{V}_L}\frac{(f,v_L)_{\V\times\V^*}}{||v_L||_a},$ cf. \eqref{eq:DualNorm}
and where $r^m$, $m=0,\ldots,M-1$ denote the weak residuals defined in equation \eqref{eq:ErrorThetaScheme}. 
\end{corollary}
The weak residual $r^m$, $m=0,\ldots,M-1$ in \eqref{eq:ErrorThetaScheme} can be decomposed into the following parts
\begin{align*}
(r^m,v_L)_{\V \times \V^*}:=(r^m_1,v_L)_{\V \times \V^*}+(r^m_2,v_L)_{\V \times \V^*}+a(r^m_3,v_L),
\end{align*}
where the components $r_1^m$, $r_2^m$ and $r_3^m$, $m=0,\ldots,M$ are defined as
\begin{align}\label{eq:ErrorDecompResiduals}\begin{split}
&(r_1^m,v_L)_{\V \times \V^*}:=(\tfrac{1}{k}(u^{m+1}-u^m)-\dot{u}^{m+\theta},v_L)_{\V \times \V^*},\\
&(r_2^m,v_L)_{\V \times \V^*}:=(\tfrac{1}{k}(P_Lu^{m+1}-P_Lu^m)+\tfrac{1}{k}(u^{m+1}-u^m),v_L)_{\V \times \V^*},\\
&(r_3^m,v_L)_{\V \times \V^*}:=a(P_Lu^{m+\theta}-u^{m+\theta},v_L).
\end{split}\end{align}
Using this decomposition facilitates the following estimates for the residuals.
\begin{lemma}[Norm estimates for the residuals]\label{Lem:ErrorResidualEstimates}
Consider the weak residuals $r^m$ of the $\theta$-scheme \eqref{eq:ErrorThetaScheme} for $m=0,\ldots,M-1$.
Furthermore, assume that 
$$u\in C^1(\bar{J};\mathcal{H}^2_j(G,x^{\mu/2}))\cap C^3(J;\mathcal{H}^2_j(G,x^{\mu/2})), \qquad j=0,1,$$ 
where $\mathcal{H}^k_j(G,x^{\mu/2})$, $k=0,1,2$, $j=0,1$ are the spaces in \eqref{eq:DefSpacesForApproxProp2dim}.
Then there holds the estimate
\begin{equation}\label{eq:ErrorEstimResiduals}
\begin{array}{lll}
||r^m||_* \leq & C &  \begin{cases}
               k^{1/2}\left(\int_{t_m}^{t_{m+1}}||\ddot{u}||_{*}^2ds\right)^{1/2},\quad \theta\in[0,1]\\
               k^{3/2}\left(\int_{t_m}^{t_{m+1}}||\dddot{u}||_{*}^2ds\right)^{1/2},\quad \theta=\tfrac{1}{2}\\
              \end{cases}\\
&&+ 2^{-L}\left(\tfrac{ C  }{k^{1/2}} \left(\int_{t_m}^{t_{m+1}}||\dot{u}||_{\Hh^1_{j}(G,x^{\mu/2})}^2ds\right)^{1/2}
+  \ C  ||u^{m+\theta}||_{\mathcal{H}^2_{j}(G,x^{\mu/2})}\right),
\end{array}
\end{equation}
\end{lemma}
A proof of this Lemma is provided in the Appendix \ref{Sec:Proofs}. With these preparations, we are in a position to prove the main result of this section:

\begin{theorem}[Convergence of the finite element approximation: SABR]\label{Th:FEMConvergence}
Assume that 
\begin{align*}
u\in C^1(\bar{J};\mathcal{H}^2_j(G,x^{\mu/2}))\cap C^3(J;\mathcal{H}^2_j(G,x^{\mu/2})), 
\qquad j=0,1, 
\end{align*}
where $\mathcal{H}^k_j(G,x^{\mu/2})$, $k=0,1,2$, $j=0,1$ are the spaces in \eqref{eq:DefSpacesForApproxProp2dim},
and assume further that the approximation $u_{(0,L)}\in V^L$ of the initial data is
quasi optimal, that is 
\begin{align}\label{eq:quasioptimal}
||\xi_L^0||^2_{\Hh}=||u_0-u_{(0,L)}||^2_{\Hh}\leq C 2^{-2L}||u_0||^2_{\Hh}.
\end{align}
Let $u^m(z)=u(t^m,z)$, $z\in G$ for $t^m$, $m=0,\ldots M$ be as in \eqref{eq:TimeStep} let $u_L^m$ denote the solution of the fully discrete scheme \eqref{eq:ThetaScheme}, and let the approximation space $\V^L$ be as in Section \ref{Sec:Discretization}.
Then, the following error bounds hold:
\begin{align*}
||u^M-u^M_L||^{2}_{\Hh} + k \sum_{m=0}^{M-1}||u^{m+\theta}-u_L^{m+\theta}||_a^2
\leq  &C 2^{-j(1-\beta)2L   } \max_{0\leq t \leq T}||u(t)||^{2}_{\mathcal{H}^2_j}
+2^{-j(1-\beta)2L  } \int_0^T||\dot{u}(s)||_{\mathcal{H}^1_j}^2ds\\
&+C \begin{cases}
     k^2 \int_0^T|| \ddot{u}(s)||_{*}^2ds,\quad 0\leq \theta \leq 1\\
     k^4 \int_0^T|| \dddot{u}(s)||_{*}^2ds,\quad \theta=\frac{1}{2}
    \end{cases}
\end{align*}
for $j=0,1$, where $u_L^{m+\theta}=\theta u_L^{m+1}+(1-\theta)u_L^m$, and $u^{m+\theta}=\theta u^{m+1}+(1-\theta)u^m$.
\end{theorem}

\begin{remark}[Convergence of the finite element approximation: CEV]\label{Cor:FEMConvergence}
The same estimates hold for the CEV model, when we replace in Theorem \ref{Th:FEMConvergence} the spaces $\mathcal{H}^k_j(G,x^{\mu/2})$, $k=0,1,2$, $j=0,1$ and the corresponding norms by $H^{k}_j(I,x^{\mu/2})$, $k=0,1,2$, $j=0,1$ in \eqref{eq:DefSpacesForApproxProp} and the norms \eqref{eq:DefNormSpacesForApproxProp}.
\end{remark}

\begin{proof}[Proof of Theorem \ref{Th:FEMConvergence}, and Remark \ref{Cor:FEMConvergence}]For the proofs we follow \cite[Theorem 3.6.5]{ReichmannComputationalMethods}, and \cite[Theorem 5.4.]{PetersdorffSchwab} with the appropriate adjustments. For brevity we shall prove both cases $\Hh_{j=0}$ and $\Hh_{j=1}$ together, and denote the generic convergence of order by $2^{-2L c_{\omega}}:=2^{-2L - 2L j(1-\beta)}$ as in \eqref{eq:ApproximationPropertyWeighted}, where $c_{\omega}=1$ for $\Hh_{j=0}$ and  $c_{\omega}=(1-\beta)$ for $\Hh_{j=1}$.
By Corollary \ref{Cor:StabilityErrorThetaScheme},
\begin{align*}
||e_L^M||^2_{\Hh}&=||u^m(x)-u_L^m(x)||^2_{\Hh}=||\eta^m+\xi_L^m||^2_{\Hh}\\
&\leq 2 \Big(||\eta^m||^2_{\Hh}+k \sum_{m=0}^{M-1}||\eta^{m+\theta}||_a^2
+||\xi_L^m||^2_{\Hh}+k \sum_{m=0}^{M-1}||\xi_L^{m+\theta}||_a^2 \Big).
\end{align*}
This yields
\begin{align*}
||e_L^M||^2_{\Hh}+C_1  k \sum_{m=0}^{M-1}||e^{m+\theta}||_a^2
&\leq 2 \Big(||\eta^M||^2_{\Hh}+C_1 \ k \sum_{m=0}^{M-1}||\eta^{m+\theta}||_a^2
+||\xi_L^m||^2_{\Hh}+C_1 \ k \sum_{m=0}^{M-1}||\xi_L^{m+\theta}||_a^2 \Big)\\
& 
\leq C \Big(||\eta^M||^2_{\Hh}+k \sum_{m=0}^{M-1}||\eta^{m+\theta}||_a^2
+||\xi_L^0||^2_{\Hh}+C_2 \ k \sum_{m=0}^{M-1}||r^{m}||_{*}^2 \Big),
\end{align*}
and by Lemma \ref{Lem:ErrorResidualEstimates} one can further estimate the last terms to obtain
\begin{align*}
&C \Big(||\eta^M||^2_{\Hh}+k \sum_{m=0}^{M-1}||\eta^{m+\theta}||_a^2
+||\xi_L^0||^2_{\Hh}+C_2 \ k \sum_{m=0}^{M-1}||r^{m}||_{*}^2 \Big)\\
&\leq C\Bigg\{||\eta^M||^2_{\Hh}+k \sum_{m=0}^{M-1}||\eta^{m+\theta}||_a^2
+||\xi_L^0||^2_{\Hh}\\
&+C_3  \sum_{m=0}^{M-1}
 (2^{-L})^{2 c_{\omega}} \left(\int_{t_m}^{t_{m+1}}||\dot{u}||_{\mathcal{H}^1_{j}}^2ds
+  ||u^{m+\theta}||_{\mathcal{H}^2_{j}}^2
+  \begin{cases}
               k^2\int_{t_m}^{t_{m+1}}||\ddot{u}||_{*}^2ds,\quad \theta\in[0,1]\\
               k^4\int_{t_m}^{t_{m+1}}||\dddot{u}||_{*}^2ds,\quad \theta=\tfrac{1}{2}\\
   \end{cases}\right)\Bigg\}\\
&\leq C ||\xi_L^0||^2_{\Hh}+ C  
 (2^{-L})^{2 c_{\omega}}\int_{0}^{T}||\dot{u}||_{\V}^2ds
+ C \max_{0\leq t\leq T} (2^{-L})^{2 c_{\omega}}||u(t)||_{\mathcal{H}^2_{j}}^2
+ C \begin{cases}
               k^2\int_{0}^{T}||\ddot{u}||_{*}^2ds,\ \theta\in[0,1]\\
               k^4\int_{0}^{T}||\dddot{u}||_{*}^2ds,\ \ \theta=\tfrac{1}{2},
   \end{cases}
\end{align*}
where the last step follows by $||\cdot||_a\leq||\cdot||_{\V}\leq C_G||\cdot||_{\mathcal{H}^1_{j}(G,x^{\mu/2})}$ for a $C_G>0$ (cf. Remark \ref{Rem:CoincidingSpacesEstimates2dim}) by Lemma \ref{Th1:WellPosednessSABR} and the approximation estimates \eqref{ApproximationPropertyInteger} resp. \eqref{ApproximationPropertyIntegerBetter}.
Finally, quasi optimality \eqref{eq:quasioptimal} of the initial data yields the statement of the Theorem.
\end{proof}
\section{Numerical Experiments}
In this section we carry out numerical experiments for option prices to investigate the robustness of the derived finite element discretisation in a simple linear setup described in \cite[Chapter 4]{ReichmannComputationalMethods}. We find that 
the method performs robustly even for long maturities (which is precisely where the short-time asymptotic formula of Hagan et al \cite{HLW} is known to break down), for different regimes of the CEV parameter (both for $\beta<0.5$ and for $\beta \geq 0.5$) and throughout correlations. We display two examples here and remark that the convergence we find in practice (readily for this bases) outperforms the theoretically predicted ones.\\
\begin{center}\includegraphics[scale=0.21]{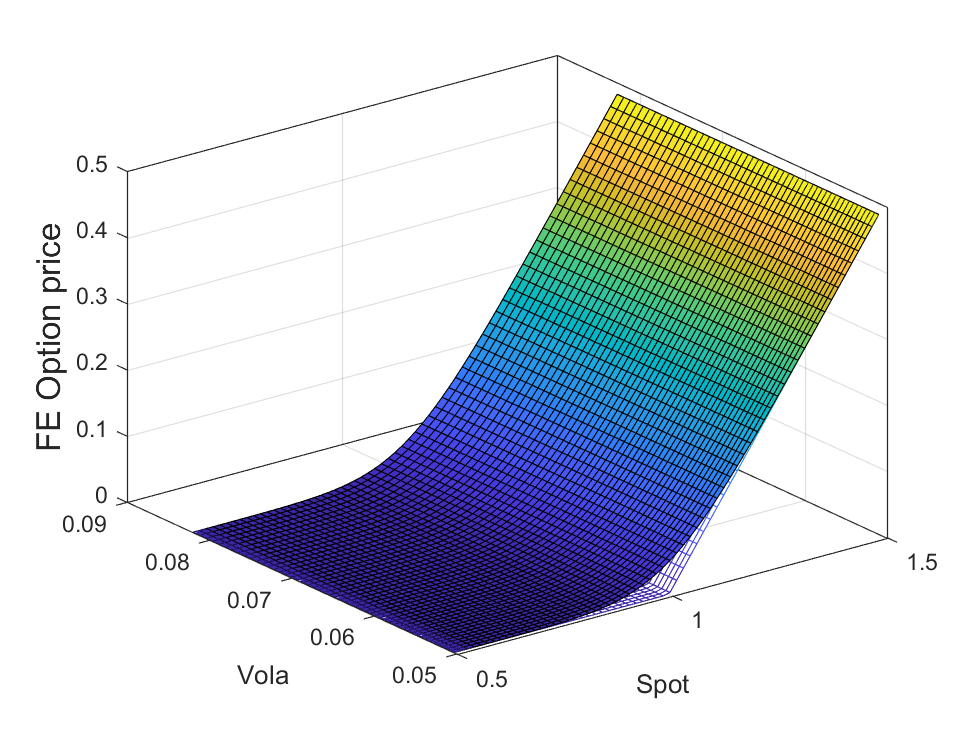}\quad
\includegraphics[scale=0.2]{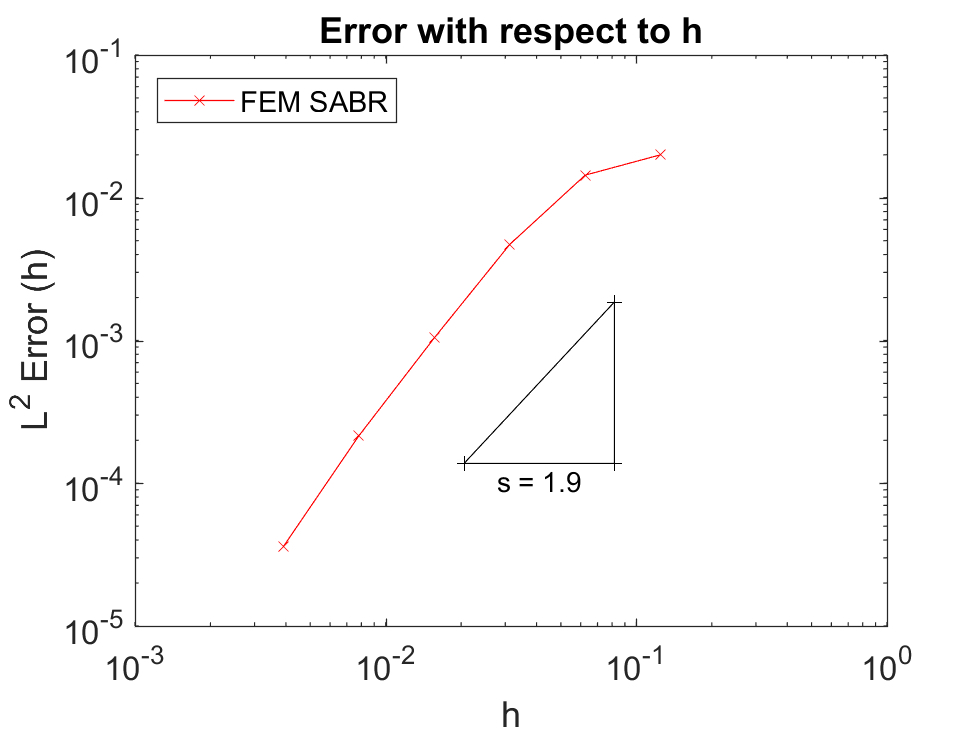}
\end{center}
The image on the left gives finite element option prices for $\beta=0.2$ and $\nu=1$ in the uncorrelated case, with maturity $T=25$ years and $K=1$, while the image on the right shows the convergence (with respect to the mesh-width $h=2^L$) of the finite element approximation in a linear basis.\\
\begin{center}
\includegraphics[scale=0.23]{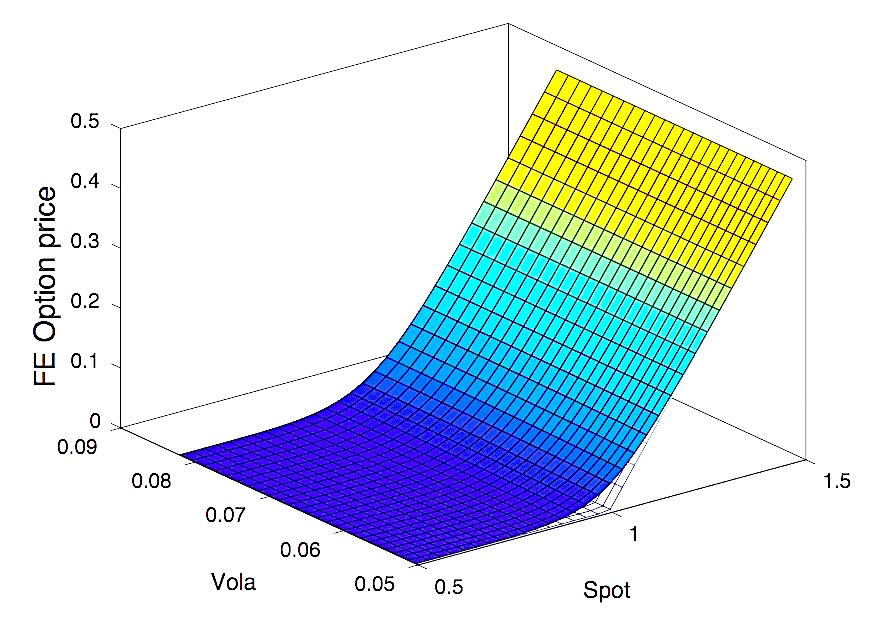} \quad 
\includegraphics[scale=0.2]{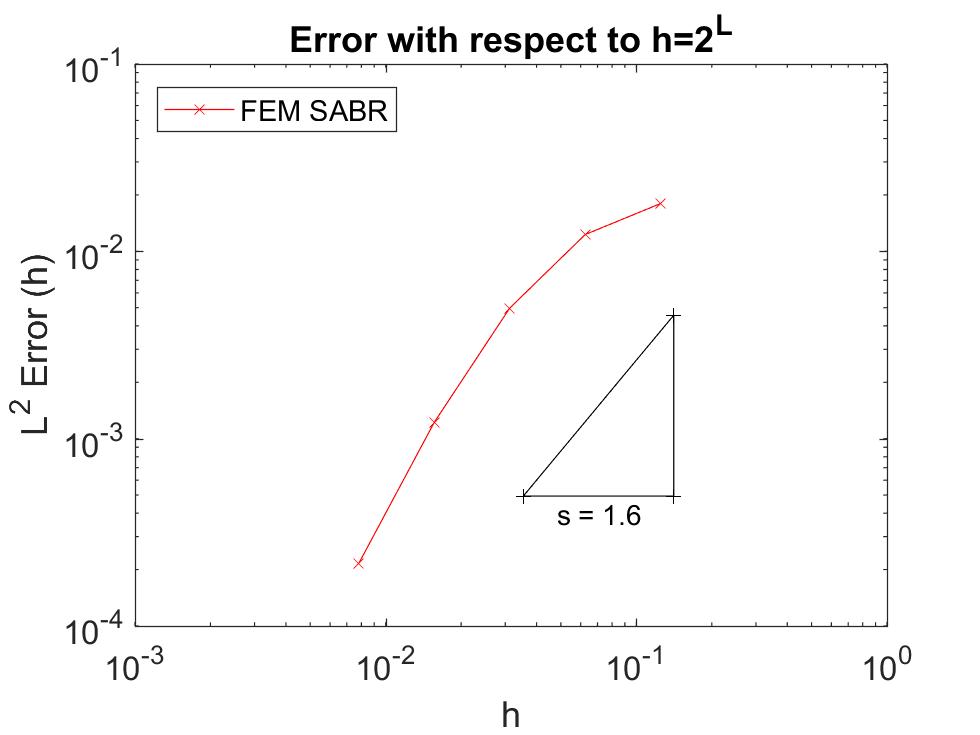}
\end{center}
Here, the image on the left gives finite element option prices for $\beta=0.5$, $\nu=1$ and correlation $\rho=-0.3$ with maturity $T=10$ years, while the image on the right shows the convergence (with respect to the mesh-width $h=2^L$) of the finite element approximation in a linear basis.\\
\begin{center}
\includegraphics[scale=0.22]{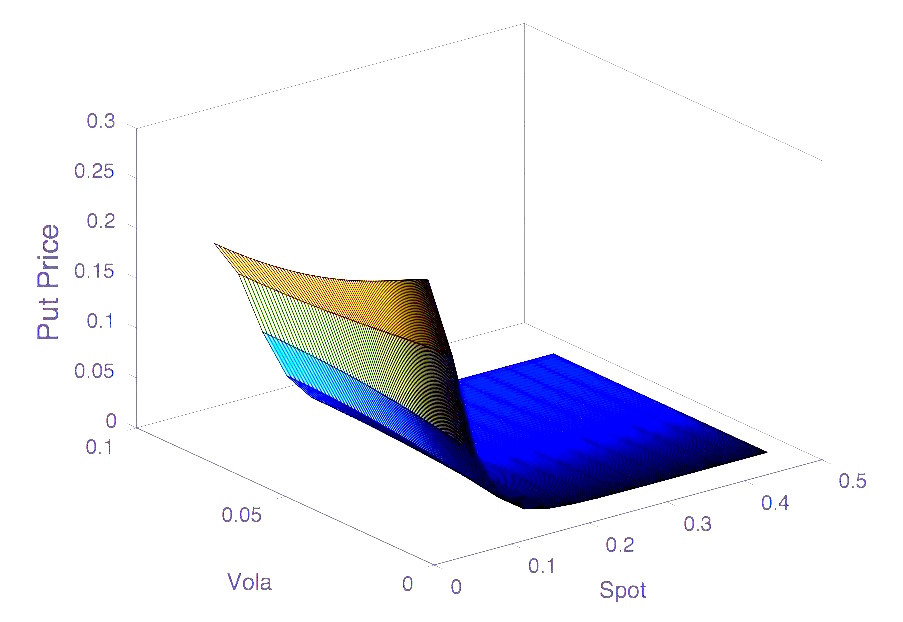}
\includegraphics[scale=0.21]{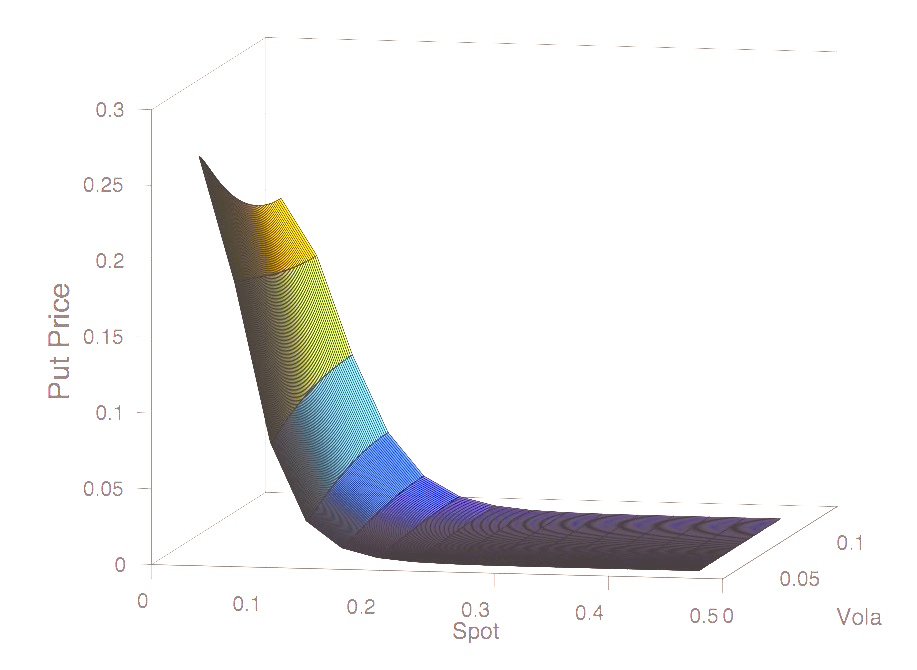}
\end{center}
The last images  display put options as an approximation of the mass at zero on the time horizon $T=10$. The put options are of the form $\max(1-\tfrac{1}{\eps}X,0)$ for a small $\epsilon>0$. The parameters in this approximation are $\beta=0.2,$ $\rho=0$, and $\nu=1$, and we chose throughout $\mu=-\beta$. 
\\

\appendix
\section{Reminder on properties of considered function spaces}
\subsection{Weighted spaces}\label{Sec:WeightedSobolevSpaces}
To make the reading self-contained, we include a reminder on some properties of weighted Sobolev spaces used in this article
and refer the reader to the monographs of \cite{Kufner} and \cite{Cavalheiro} for full details.
By a weight, we shall mean a locally integrable function $\omega$ on $\R^2$ such that
$\omega(x) > 0$ a.e. Every weight $\omega$ gives rise to a measure (via integration $\omega(E)=\int_E \omega(x)dx$, for measurable sets $E\subset R^2$). This measure is also denoted by $\omega$.
\begin{definition}[Weighted $\Ll^2$-space]\label{Def:WeightedL2Space}
Let $\omega$ be a weight on an open set $G \subset \R^2$. $\Ll^2(G, \omega)$ is the set of measurable functions $u$ on $G$ such that
\begin{equation}\label{eq:DefWeightedL2Space}
||u||_{\Ll^2(G,\omega)}^2=\int_G |u(x)|^2 \omega(x)dx<\infty
\end{equation}
\end{definition}
\begin{definition}[Weighted Sobolev space]\label{Def:WeightedSobolevSpace}
Let $k \in \mathbb{N}$ and Let $\bold{a} = \{\omega_{\bold{a}} = \omega_{\bold{a}}(x), x \in G, |\bold{a}| \leq k\}$ be a given family of weight functions on an open set $G \subset \R^2$. We denote by $W^k(G, \varphi)$
the set of all functions $u \in \Ll^2 (G, \omega)$ for which the weak derivatives $D^{(\bold{a})} u$, with
$| \bold{a}| \leq k$, belong to $\Ll^2 (G, \omega_{\bold{a}})$.
The weighted Sobolev space $W^k(G, \varphi)$ is a normed
linear space if equipped with the norm
\begin{equation}\label{eq:DefWeightedSobolevSpace1}
||u||_{W^k(G, \omega)}^2
=\sum_{|\bold{a}| \leq k}\int_G |D^{\bold{a}}u(x)|^2 \omega_{\bold{a}}(x)dx.
\end{equation}
\end{definition}
\begin{remark}\label{Rem:L1loc}
If $\omega_{\omega} \in \Ll^1_{loc}(G)$ then $\C_0^{\infty}(G)$
is a subset of $W^k(G, \omega)$, and we can introduce the space $W_0^k(G, \omega)$ as the closure
of $\C_0^{\infty} (G)$ with respect to the norm $W^k(G, \omega)$, see also \cite{Cavalheiro,KufnerOpic}.
\end{remark}
\noindent The class of $A^p$ weights was introduced by B. Muckenhoupt (cf. \cite{Muckenhoupt}). This class is relevant for the property that it ensures $\C_0^{\infty}(G) \subset W^k(G, \omega)$. This is used in Sections \ref{Sec:SettingandVariationalFormulation} and \ref{Sec:WellPosednessSABR}.\\
A weight $\omega$ is in $A^p$ if 
there exists a positive constant
$C$ such that for every ball $B \subset \R^2$
\begin{equation}
\left(\frac{1}{|B|}\int_B \omega dx\right) \left(\frac{1}{|B|} \int_B \omega^{-1/(p-1)} dx\right)^{p-1}\leq C.
\end{equation}
\begin{lemma}\label{Lem:MuckenhouptWeights}
$\omega(x) := |x|$, $x\in \R^2$  is in $A_p$ if and only if $-2 < \omega < 2(p - 1)$.
\end{lemma}
\begin{proof}
See Corollary 4.4 in \cite{Torchinsky}, and Corollary 2.18 in \cite{GarciaCuervaFrancia}.
\end{proof}

\begin{lemma}\label{Lem:ClosednessContinuousEmbedding}
If $\omega \in A^p$ then since $\omega^{-1/(p-1)}$ is locally integrable, we have $L^p(G, \omega) \subset L^1_{loc}(G)$ for
every open set $G\subset \R^n$. Therefore, weak derivatives of functions in
$L^p(G, \omega)$ are well-defined. Furthermore, if $\omega \in A^p$ then $\C^{\infty} (G)$ is dense in $W^{k,p} (G, \omega)$.
\end{lemma}
\begin{proof}
See \cite[Corollary 2.1.6 ]{Turesson} and \cite[Theorem 1.5]{FabesKenigSerapioni}. 
\end{proof}

\subsection{Tensor Products of Hilbert Spaces}\label{Sec:TensorHilbertSpaces}
Let $I:=(0,1)$ denote the unit interval and $G:=I\times I$. See in \cite[Section 13.1]{ReichmannComputationalMethods} that the Hilbert spaces
$H^k(G)$, $k=0,1,2$ can be constructed from $H^k(I)$ via the tensor product structure:
\begin{align}
\mathcal{L}^2(G)\cong&\left(L^2(I)\otimes L^2(I)\right)\label{eq:Tensorspacek=0}
\\
\Hh^1(G)\cong&\left(H^1(I)\otimes L^2(I)\right)\bigcap\left(L^2(I)\otimes H^1(I)\right)\label{eq:Tensorspacek=1}
\\
\Hh^2(G)\cong&\left(H^2(I)\otimes L^2(I)\right)\bigcap\left(H^1(I)\otimes H^1(I)\right)\bigcap\left(L^2(I)\otimes H^2(I)\right)\label{eq:Tensorspacek=2}
\end{align}
Recall from \cite[Chapter II. 4]{ReedSimon} that the inner products on each of the tensor-Hilbert spaces are defined as
\begin{equation}
\langle u_1 \otimes u_2,  v_1 \otimes v_2  \rangle _{H_1\otimes H_2} := \langle u_1, v_1 \rangle _{H_1} \langle u_2, v_2 \rangle  _{H_2},
\end{equation} for $u_1,v_1\in H_1$ and $u_2,v_2\in H_2$, where $H_1$, $H_2$ stand for generic Hilbert spaces (say any of the tensor products involved in \eqref{eq:Tensorspacek=0},\eqref{eq:Tensorspacek=1}, or \eqref{eq:Tensorspacek=2} above).
The inner products on the intersection spaces $\Hh^k(G)$, $k=0,1,2$ in \eqref{eq:Tensorspacek=0},\eqref{eq:Tensorspacek=1} and \eqref{eq:Tensorspacek=2} induced by this construction are equivalent to the usual norms on these
spaces. \\
Furthermore, to justify $u\equiv u_x\otimes u_y $, and $v\equiv v_x\otimes v_y \in \mathcal{L}^2(G)$ for $u_x,v_x\in L^2(I)$ and $u_y,v_y\in L^2(I)$ we recall the following \cite[Theorem II. 10. c), Chapter II. 4]{ReedSimon}:
\begin{theorem}\label{Th:DecompositionL2functions}
Let $(M_1, \mu_1)$ and $(M_2, \mu_2)$ be measure spaces so that $L^2(M_1, \mu_1)$ and $L^2(M_2, \mu_2)$ are separable, then
there is a unique isomorphism such that 
\begin{equation}\begin{array}{rl}
L^2(M_1 \times M_1, d\mu_1 \otimes d\mu_2)&\longmapsto L^2(M_1, d\mu_1; L^2(M_2, d\mu_2))\\
f(x,y)&\longmapsto (x\mapsto f(x,\cdot)).
\end{array}\end{equation}
\end{theorem}
\subsubsection{Explicit construction of the bivariate spaces in Section \ref{Sec:ApproximationEstimatesWeighted}}\label{Sec:TensorHilbertSpacesExplicit}
Consider our domain of interest $G=(0,R_x)\times(-R_y,R_y)$, $R_x,R_y>0$ and the weighted Sobolev spaces
\begin{equation}\label{eq:DefSpacesForApproxProp2dim}
\mathcal{H}^k_j(G,x^{\mu/2}):=\{u:G\rightarrow \R \ \textrm{measurable}: \  \partial^{|\bold{a}|}_{\bold{a}} u \in L^2(I, x^{\mu/2+a_x \beta j}),  \ \bold{a} \leq k \},\ k=0,1,2,
\end{equation}
for $j=0,1$, and a multiindex $\bold{a}$ with $|\bold{a}|=a_x+a_y$, where $a_x$ denotes the number of derivatives in direction $x$ and $a_y$ in direction $y$.
The respective norms in $\mathcal{H}^k_j(G,x^{\mu/2})$ for $j=0$ are defined by
\begin{equation}\label{eq:SpacesForApproxProp2dimj0}\begin{array}{ll}
||u||_{\mathcal{H}^0_{j=0}(G, x^{\mu/2})}^2:=&||x^{\mu/2}u ||_{\Ll^2}^2\\
||u||_{\mathcal{H}^1_{j=0}(G, x^{\mu/2})}^2:=&||x^{\mu/2}u ||_{\Ll^2}^2+||x^{\mu/2} \partial_y(u)||_{\Ll^2}^2+||x^{\mu/2} \partial_x(u) ||_{\Ll^2}^2\\
||u||_{\mathcal{H}^2_{j=0}(G, x^{\mu/2})}^2:=&||x^{\mu/2}u ||_{\Ll^2}^2+||x^{\mu/2} \partial_y(u)||_{\Ll^2}^2+||x^{\mu/2} \partial_x(u) ||_{\Ll^2}^2\\
&+ ||x^{\mu/2}\partial_{yy}u ||_{\Ll^2}^2+||x^{\mu/2}\partial_{xy}u ||_{\Ll^2}^2+||x^{\mu/2}\partial_{yy}u ||_{\Ll^2}^2,
\end{array}
\end{equation}
and for $j=1$ by
\begin{equation}\label{eq:SpacesForApproxProp2dimj1}\begin{array}{ll}
||u||_{\mathcal{H}^0_{j=1}(G, x^{\mu/2})}^2:=&||x^{\mu/2}u ||_{\Ll^2}^2\\
||u||_{\mathcal{H}^1_{j=1}(G, x^{\mu/2})}^2:=&||x^{\mu/2}u ||_{\Ll^2}^2+||x^{\mu/2} \partial_y(u)||_{\Ll^2}^2+||x^{\beta+\mu/2} \partial_x(u) ||_{\Ll^2}^2\\
||u||_{\mathcal{H}^2_{j=1}(G, x^{\mu/2})}^2:=&||x^{\mu/2}u ||_{\Ll^2}^2+||x^{\mu/2} \partial_y(u)||_{\Ll^2}^2+||x^{\beta+\mu/2} \partial_x(u) ||_{\Ll^2}^2\\
&+ ||x^{\mu/2}\partial_{yy}u ||_{\Ll^2}^2+||x^{\beta+\mu/2}\partial_{xy}u ||_{L^2}^2+||x^{2\beta+\mu/2}\partial_{yy}u ||_{L^2}^2.
\end{array}
\end{equation}

\begin{remark}\label{Rem:CoincidingSpacesEstimates2dim}
Note, that similarly as in Remark \ref{Rem:CoincidingSpacesEstimates}, the spaces $\Hh$ and $\V$ in Definitions \ref{Def:HilbertSpaceSABR} and \ref{Def:WeightedSobolevSpaceVSABR}
and the spaces $\Hh^k_j(G)$, $k=0,1,2$, $j=0,1$ with norms as in \eqref{eq:SpacesForApproxProp2dimj0} and \eqref{eq:SpacesForApproxProp2dimj1}
 coincide $\Hh=\Hh^{0}_{j=0}(G,x^{\mu/2})=\Hh^{0}_{j=1}(G,x^{\mu/2})$, and the weighted space $\V$ satisfies $\V=\Hh^1_{j=1}(G, x^{\mu/2})$ and $\V \supset \Hh^1_{j=0}(G, x^{\mu/2})$, furthermore, on any bounded domain $G$ there holds the estimate 
\begin{equation}\label{eq:EstimateVvsWeightedH1}
||v||_{\V}^2\leq C_G ||v||^2_{\mathcal{H}^1_{j=0}(G, x^{\mu/2})} , \quad v\in \V, \quad C_G>0. 
\end{equation}
\end{remark}

\begin{lemma}\label{Lem:TensorHilbertSpacesExplicit}
The spaces in \eqref{eq:DefSpacesForApproxProp2dim}, $k=0,1,2$, $j=0,1$ can be constructed as tensor products of the spaces \eqref{eq:DefSpacesForApproxProp} and the usual (unweighted) Sobolev spaces $H^k(G)$, $k=0,1,2$, $j=0,1$ via \eqref{eq:Tensorspacek=0} \eqref{eq:Tensorspacek=1} and \eqref{eq:Tensorspacek=2} as follows
\begin{align*}
\mathcal{H}^0_{j}(G, x^{\mu/2})\cong&\left(H^0_j(I, x^{\mu/2})\otimes H^0(I)\right)
\\
\mathcal{H}^1_j(G,x^{\mu/2})\cong&\left(H^1_j(I, x^{\mu/2})\otimes H^0(I)\right)\bigcap\left(H^0_j(I, x^{\mu/2})\otimes H^1(I)\right)
\\
\mathcal{H}^2_j(G,x^{\mu/2})\cong&\left(H^2_j(I, x^{\mu/2})\otimes H^0(I)\right)\bigcap\left(H^1_j(I, x^{\mu/2})\otimes H^1(I)\right)\bigcap\left(H^0_j(I, x^{\mu/2})\otimes H^2(I)\right).  
\end{align*} 
\end{lemma}


\section{Non-symmetric Dirichlet forms}\label{Sec:NonSymmDirichlet} 
We include here a condensed reminder of some of the basic concepts on non-symmetric Dirichlet forms, which are used in previous sections, in particular in Theorem \ref{Th:SABRDirichletForm}. For full details see \cite{RoecknerMa}.
\begin{definition}[Symmetric closed form]\label{Def:SymmetricClosedForm}
A pair $(\cE,D(\cE))$ is called a \emph{symmetric closed form} on the Hilbert space $(\Hh,(\cdot,\cdot)_{\Hh})$, if $D(\cE)$ is a dense linear subspace  of$\Hh$,
and $\cE:D(\cE)\times D(\cE)\rightarrow \R$ is a non-negative definite symmetric bilinear form, which is closed on $\Hh$. That is, $D(\cE)$ is a complete metric space
with respect to the norm $\cE_1(\cdot,\cdot)^{1/2}:=(\cE(\cdot,\cdot)+(\cdot,\cdot)_{\Hh})^{1/2}$.
\end{definition}
\begin{definition}[Coercive closed form]\label{Def:CoerciveClosedForm}
A pair $(\cE,D(\cE))$ is called a \emph{coercive closed form} on the Hilbert space $\Hh$ if $D(\cE)$ is a dense linear subspace of $\Hh$, and
$\cE:D(\cE)\times D(\cE)\rightarrow \R$ is a bilinear form such that the following two conditions hold:
\begin{itemize}
 \item Its \emph{symmetric part $\widetilde{\cE}(u,v):=\frac{1}{2}\left(\cE(u,v)+\cE(v,u)\right)$} is a symmetric closed form on $\Hh$.
 \item The pair $(\cE,D(\cE))$ satisfies the so-called \emph{weak sector condition}: there exists a \emph{continuity constant} $K>0$  such that 
\begin{equation}\label{eq:weaksector}
 |\cE_1(u,v)|\leq K \ \cE_1(u,u)^{1/2}\cE_1(v,v)^{1/2}\quad \textrm{for all} \ u,v \in D(\cE).
\end{equation}
\end{itemize}
\end{definition}
\begin{remark}\label{Rem:StrongSectorWeakSector}
Recall the continuity property \eqref{eq:Defcontinuity} (also considered in Section \ref{Sec:WellPosednessSABR}) 
\begin{equation*}
 |\cE(u,v)|\leq K \ \cE(u,u)^{1/2}\cE(v,v)^{1/2}\quad \textrm{for all} \ u,v \in D(\cE)
\end{equation*}
implies the weak sector condition \eqref{eq:weaksector} above.
\end{remark}

\begin{definition}[Dirichlet form]\label{Def:NonSymmDirichlet}
Consider a Hilbert space $(\Hh,(\cdot,\cdot)_{\Hh})$ of the form $\Hh~=~\Ll^2(E,m)$, where 
$(E,m)$ is a measure space.
A coercive closed form $(\cE,D(\cE))$ on $\Ll(E,m)$ is called a \emph{(non-symmetric) Dirichlet form}, if for all $u \in \DcE\subset E$, one has the \emph{contraction properties}
\begin{equation}\label{eq:SubMContrBilinearformNonSymm}
\begin{array}{lll}
u^+\wedge 1 \in \DcE &\textrm{and}&\cE(u+u^+\wedge 1,u-u^+\wedge 1)\geq 0\\
&\textrm{and}& \cE(u-u^+\wedge 1,u+u^+\wedge 1)\geq 0,
\end{array}
\end{equation}
where for any $u,v:E\rightarrow \R$, we have set
\begin{equation}\label{eq:WedgePositivePart}
\begin{array}{llll}
u\wedge v:=\inf(u,v),& u \vee v:=\sup(u,v),&u^+:=u \vee 0, &u^-:=-(u\wedge 0).
\end{array}
\end{equation}
A coercive closed form satisfying one of the two inequalities in \eqref{eq:SubMContrBilinearformNonSymm} is called \emph{$\frac{1}{2}$-~Dirichlet~form}. \\
If $(\cE, \DcE)$ is in addition symmetric, that is $\cE=\widetilde{\cE}$, where $\widetilde{\cE}$ denotes the symmetric part of $\cE$ (recall $\widetilde{\cE}(u,v):=\frac{1}{2}\left(\cE(u,v)+\cE(v,u)\right)$), then 
$(\cE,\DcE)$ is called a \emph{symmetric} Dirichlet form. \\In the latter case, the contraction property in condition \eqref{eq:SubMContrBilinearformNonSymm} reduces to
\begin{equation}\label{eq:SubMContrBilinearformSymm}
\cE(u^+\wedge 1,u^+\wedge 1)\leq \cE(u,u).
\end{equation}
\end{definition}
See \cite[Section 4, Def. 4.5]{RoecknerMa}
\begin{theorem}\label{Th:AssociatedSemigroup}
Let $(\cE,D(\cE))$ be a coercive closed form on a Hilbert space $(\Hh,(\cdot,\cdot)_{\Hh})$ with continuity constant $K>0$. Define the domain
\begin{equation}\label{eq:GeneratorofDirichletform}
D(A):=\{u \in \DcE \ | \ v\mapsto \cE(u,v)\ \textrm{is continuous w.r.t. } (\cdot,\cdot)_{\Hh}^{1/2} \textrm{ on } \DcE \}. 
\end{equation}
For any $u\in D(A)$, let $Au$ denote the unique element in $\Hh$ such that 
\begin{equation}\label{eq:BilinearGenerator}
(-Au,v)=\cE(u,v) \ \textrm{for all} \ v \in \DcE. 
\end{equation}
Then $A$ is the generator of the unique strongly continuous contraction resolvent\footnote{See \cite[Section 4]{RoecknerMa}.} $(G_{\alpha})_{\alpha>0}$ on $\Hh$ which satisfies 
\begin{equation}
\begin{array}{rcl}
G_{\alpha}(\Hh), \subset \DcE&\textrm{and}&\cE(G_{\alpha}f,u)+\alpha(G_{\alpha}f,u)_{\Hh}=(f,u)_{\Hh}\\
\textrm{for all } &f\in \Hh, & u\in \DcE, \quad \alpha>0.
\end{array}
\end{equation}
Furthermore, since $(\cE,D(\cE))$ is a coercive closed form on $(\Hh,(\cdot,\cdot)_{\Hh})$, there exists a further unique strongly continuous contraction resolvent $(\hat G_{\alpha})_{\alpha>0}$ on $\Hh$, which satisfies
\begin{equation}
\begin{array}{rcl}
\hat G_{\alpha}(\Hh)\subset \DcE&\textrm{and}&(f,u)_{\Hh}= \cE(u,\hat G_{\alpha}f)+\alpha(u,\hat G_{\alpha}f)_{\Hh}\\
\textrm{for all } &f\in \Hh, & u\in \DcE, \quad \alpha>0.
\end{array}
\end{equation}
In particular, $\hat G_{\alpha}$ is the adjoint of $G_{\alpha}$ for all $\alpha>0$. That is
\begin{equation}
\begin{array}{rll}
(G_{\alpha}f,g)_{\Hh}=(f,\hat G_{\alpha}g)_{\Hh} &\textrm{for all}& f,g \in \Hh,
\end{array}
\end{equation}
and similarly, for the (unique) strongly continuous contraction semigroups $(P_t)_{t\geq0}, (\hat P_t)_{t\geq0}$ corresponding to
$(G_{\alpha})_{\alpha>0}$ and $(\hat G_{\alpha})_{\alpha>0}$ respectively it holds that
\begin{equation}
\begin{array}{rll}
(P_{t}f,g)_{\Hh}=(f,\hat P_{t}g)_{\Hh} &\textrm{for all}& f,g \in \Hh,\ t\geq 0.
\end{array}
\end{equation}
\end{theorem}

\begin{proof}
See: \cite[Theorem 2.8, Corollary 2.10 and Proposition 2.16]{RoecknerMa}. 
\end{proof}
\begin{definition}[Contraction Properties]\label{Def:SubMarkovianResSemigr}
Let $(\Hh,(\cdot,\cdot)_{\Hh})$ be a Hilbert space where $\Hh=\Ll^2(E,m)$, and where 
$(E,m)$ is a measure space.
For any $f,g \in \Ll^2(E,m)$ we write $f\leq g$ or $f<g$ for any $m$-classes $f,g$ of functions on $E$, if the respective inequality holds $m$-a.e. for corresponding representatives.
\begin{itemize}
 \item[$(i)$] Let $G$ be a bounded linear operator on $\Ll^2(E,m)$ with domain $D(G)=\Ll^2(E,m)$. Then $G$ is called 
\emph{sub-Markovian}, if for all $f \in \Ll^2(E,m)$ the condition $0\leq f\leq 1$ implies 
$0\leq Gf \leq 1$.
 \item[$(ii)$] A strongly continuous contraction semigroup $(P_t)_{t\geq 0}$ resp. resolvent $(G_{\alpha})_{\alpha>0}$
is called \emph{sub-Markovian} if all $P_t$, $t\geq 0$ resp. $\alpha G_{\alpha}$, $\alpha>0$ are sub-Markovian.
 \item[$(iii)$] A closed, densely defined operator $A$ on $(\Ll^2(E,m),(\cdot,\cdot)_{\Hh})$ is called \emph{Dirichlet operator} if
$(Au,(u-1)^+)_{\Hh}\leq 0$ for all $u\in D(A)\subset E$.
\end{itemize}
See \cite[Section 4, Def. 4.1]{RoecknerMa}. 
\end{definition}

\begin{theorem}\label{Th:EquivalenceContractionProperties}
Consider a Hilbert space $(\Hh,(\cdot,\cdot)_{\Hh})$ of the form $\Hh~=~\Ll^2(E,m)$, where $(E,m)$ is a measure space.
Let $(G_{\alpha})_{\alpha>0}$ be a strongly continuous contraction resolvent on $(\Ll^2(E,m),(\cdot,\cdot)_{\Hh})$ with corresponding generator $A$ and semigroup $(P_t)_{t\geq 0}$. 
Furthermore, let $(\cE,D(\cE))$ be a coercive closed form on $\Ll^2(E,m)$ with continuity constant $K>0$ and corresponding
resolvent $(G_{\alpha})_{\alpha>0}$.
Then the following are equivalent:
\begin{itemize}
 \item[$(i)$] $(P_t)_{t\geq 0}$ is sub-Markovian.
 \item[$(ii)$] $A$ is a Dirichlet operator.
 \item[$(iii)$] $(G_{\alpha})_{\alpha>0}$ is sub-Markovian.
 \item[$(iv)$] for all $u\in \DcE$, $u^+\wedge 1 \in \DcE$ and $\cE(u+u^+\wedge 1,u-u^+\wedge 1)\geq 0,$ that is
 $(\cE,D(\cE))$ is a $\frac{1}{2}$-Dirichlet form.
\end{itemize}
If in the above statements the operators $(G_{\alpha})_{\alpha>0}$ (resp. $(P_t)_{t\geq 0}$ and $A$) are replaced by their adjoints $(\hat G_{\alpha})_{\alpha>0}$ (resp. $(\hat P_t)_{t\geq 0}$ and $\hat A$),
then the analogous equivalences hold, where in $(iv)$ the entries of $\cE$ are interchanged. Hence, if $(iii)$ (resp. $(ii)$ or $(i)$) holds both for $(G_{\alpha})_{\alpha>0}$ (resp. $A$ or $(P_t)_{t\geq 0}$) and its adjoint $(\hat G_{\alpha})_{\alpha>0}$ (resp. $\hat A$ or $(\hat P_t)_{t\geq 0}$), then the coercive closed form
$(\cE,D(\cE))$ is a (non-symmetric) Dirichlet form.
\end{theorem}

\begin{proof}
See \cite[Section 4 Proposition 4.3 and Theorem 4.4]{RoecknerMa}. 
\end{proof}

\section{Further Proofs}\label{Sec:Proofs}

\begin{proof}[Proof of Lemma \ref{Lem:ErrorThetaScheme}]
The statement of the Lemma is by the decomposition \eqref{eq:ErrorDecompResiduals} essentially a consequence of the fact that the errors $\xi$ satisfy the same PDE as the functions $u$: The variational formulation
implies
\begin{equation}\label{eq:GenVarImplies}
\begin{array}{ll} 
(\dot{u}^{m+\theta},v)_{\V \times \V^*}+a(u^{m+\theta},v)=(g^{m+\theta},v)_{\V \times \V^*} \quad \forall v\in {\V}.
\end{array}
\end{equation}
Using the definition \eqref{eq:ErrorSplitting} of $\xi$, we rewrite \eqref{eq:ErrorThetaScheme} in the form:
\begin{equation*}\label{eq:ErrorThetaScheme2}
\begin{array}{lr} 
\tfrac{1}{k}(\xi_L^{m+1}-\xi_L^{m},v_L)_{\V \times \V^*}+a(\theta \xi_L^{m+1}+(1-\theta)\xi_L^{m},v_L)=&\\
\tfrac{1}{k}((P_Lu^{m+1}-u_L^{m+1})-(P_Lu^{m}-u_L^{m}),v_L)_{\V \times \V^*}+a(P_Lu^{m+\theta}+u_L^{m+\theta},v_L)=&\\
\left(\frac{P_Lu^{m+1}-P_Lu^{m}}{k},v_L\right)_{\V \times \V^*}+a(P_Lu^{m+\theta},v_L)-\left(\frac{1}{k}(u_L^{m+1}-u_L^{m},v_L)_{\V \times \V^*}-a(u_L^{m+\theta},v_L)\right)&
\end{array}
\end{equation*}
where we used $u^{m+\theta}:= \theta u^{m+1}+(1-\theta)u^{m}$ and the linearity of the projector: $$\theta P_Lu^{m+1}-(1-\theta)P_Lu_L^{m+1}=P_Lu^{m+\theta}.$$
Furthermore, by the $\theta$-scheme \eqref{eq:ThetaScheme} for $u$, and by \eqref{eq:GenVarImplies}
\begin{equation*}\label{eq:ErrorThetaScheme3}
\begin{array}{ll} 
\left(\frac{P_Lu^{m+1}-P_Lu^{m}}{k},v_L\right)_{\V \times \V^*}+a(P_Lu^{m+\theta},v_L)-\left(\left(\frac{u_L^{m+1}-u_L^{m}}{k},v_L\right)_{\V \times \V^*}-a(u_L^{m+\theta},v_L)\right)&\\
=\left(\frac{P_Lu^{m+1}-P_Lu^{m}}{k},v_L\right)_{\V \times \V^*}+a(P_Lu^{m+\theta},v_L)-(g^{m+\theta},v_L)_{\V \times \V^*}&\\
=\left(\frac{P_Lu^{m+1}-P_Lu^{m}}{k},v_L\right)_{\V \times \V^*}+a(P_Lu^{m+\theta},v_L)-a(u^{m+\theta},v_L)-(\dot{u}^{m+\theta},v_L)_{\V \times \V^*}&\\
=\left(\tfrac{P_Lu^{m+1}-P_Lu^{m}}{k}-\tfrac{u^{m+1}-u^{m}}{k},v_L\right)_{\V \times \V^*}+a\left(P_Lu^{m+\theta}-u^{m+\theta},v_L\right)+\left(\tfrac{u^{m+1}-u^{m}}{k}-\dot{u}^{m+\theta},v_L\right)_{\V \times \V^*}&\\
=(r_2,v_L)_{\V \times \V^*}+(r_3,v_L)_{\V \times \V^*}+(r_1,v_L)_{\V \times \V^*}.&
\end{array}
\end{equation*}
\end{proof}
\begin{proof}[Proof of Lemma \ref{Lem:ErrorResidualEstimates}]
We adapt \cite[Section 3.6.2]{ReichmannComputationalMethods} and \cite[Lemma 5.3]{PetersdorffSchwab} to the situation at hand and confirm that the estimates of \cite[Lemma 5.3]{PetersdorffSchwab} carry over to the weighted case.
Analogously to the classical (unweighted) case, the statement of the Lemma follows from $||r^{m}||_{*}^2\leq||r_1^{m}||_{*}^2+||r_2^{m}||_{*}^2+||r_3^{m}||_{*}^2$ and the corresponding norm estimates for the decomposition \eqref{eq:ErrorDecompResiduals}.
The estimate of the residual $r_1$ is
\begin{equation}\label{eq:ErrorResidualr1}
\begin{array}{lrl}
||r_1||_*&=||\tfrac{1}{k}(u^{m+1}-u^m)-\dot{u}^{m+\theta}||_*&\leq\tfrac{1}{k}\left(\int_{t_m}^{t_{m+1}}|s-(1-\theta)t_{m+1}-\theta t_m|\ ||\ddot{u}||_{*}ds\right)\\
&&\leq \tfrac{C_{\theta}}{k^{1/2}}\left(\int_{t_m}^{t_{m+1}}||\ddot{u}||_{*}^2ds\right)^{1/2}. 
\end{array}
\end{equation}
In case $\theta=\frac{1}{2}$, partial integration yields the refined estimate
\begin{equation}\label{eq:ErrorResidualr1theta05}
\begin{array}{lrl}
||r_1||_*&=||\tfrac{1}{k}(u^{m+1}-u^m)-\dot{u}^{m+\theta}||_*&\leq\tfrac{1}{2k}\left(\int_{t_m}^{t_{m+1}}|(t_{m+1}-s)(t_m-s)|\ ||\dddot{u}||_{*}ds\right)\\
&&\leq \tfrac{C}{k^{3/2}}\left(\int_{t_m}^{t_{m+1}}||\dddot{u}||_{*}^2ds\right)^{1/2}. 
\end{array}
\end{equation}
The norm of the residual $r_2$ is bounded by 
\begin{equation}\label{eq:ErrorResidualr2Estim}
\begin{array}{ll}
||r_2||_*&\leq 2^{-L}\tfrac{C}{k^{1/2}}\left(\int_{t_m}^{t_{m+1}}||\dot{u}||_{\V}^2ds\right)^{1/2}.
\end{array}
\end{equation}
The bound \eqref{eq:ErrorResidualr2Estim} follows from \eqref{eq:DualNorm} and from the estimate
\begin{equation}\label{eq:ErrorResidualr2}
\begin{array}{ll}
|(r_2,v_L)_{\V\times\V^*}|&=|(\tfrac{1}{k}(P_Lu^{m+1}-P_Lu^m)+\tfrac{1}{k}(u^{m+1}-u^m),v_L)_{\V\times\V^*}|\\
&\leq C ||\tfrac{1}{k}(P_Lu^{m+1}-P_Lu^m)+\tfrac{1}{k}(u^{m+1}-u^m)||_*||v_L||_a\\
&\leq \tfrac{C}{k} ||(I-P_L)\int_{t_m}^{t_{m+1}}\dot{u}(s)\ ds||_*||v_L||_a\\
&\leq \tfrac{C}{k^{1/2}}\left(\int_{t_m}^{t_{m+1}}||\dot{u}-P_L\dot{u}||_{\Hh}^2ds\right)^{1/2}\ ||v_L||_a\\
&\stackrel{\eqref{eq:ApproximationPropertyWeighted}}{\leq} \tfrac{C}{k^{1/2}}(2^{-L})^{c_{\bold{a}}}\left(\int_{t_m}^{t_{m+1}}||\dot{u}||_{\Hh^1_j}^2ds\right)^{1/2}\ ||v_L||_a, 
\end{array}
\end{equation}
where the last step follows from the approximation property \eqref{ApproximationPropertyInteger} resp. \eqref{ApproximationPropertyIntegerBetter} and the estimate \eqref{eq:EstimateVvsWeightedH1}.
The norm of the residual $r_3$ allows for the upper bound 
\begin{equation}\label{eq:ErrorResidualr3Estim}
\begin{array}{ll}
||r_3||_*&\leq C (2^{-L})^{c_{\bold{a}}}\ ||u^{m+\theta}||_{\mathcal{H}^{2}_{j}}.
\end{array}
\end{equation}
The estimate \eqref{eq:ErrorResidualr3Estim} follows from \eqref{eq:DualNorm} and from 
\begin{equation}\label{eq:ErrorResidualr3}
\begin{array}{ll}
|(r_3,v_L)_{\V\times\V^*}|&=|a(P_Lu^{m+\theta}-u^{m+\theta},v_L)|\\
&\leq||P_Lu^{m+\theta}-u^{m+\theta}||^2_a||v_L||^2_a\stackrel{\eqref{eq:ApproximationPropertyWeighted}}{\leq} C (2^{-L})^{c_{\bold{a}}}\ ||u^{m+\theta}||^2_{\mathcal{H}^{2}_{j}}||v_L||^2_a.
\end{array}
\end{equation}
The second step follows from a simple polarisation argument
\begin{equation*}\begin{array}{ll}
\small |(u,v)_{\V\times\V^*}|^2
\leq\left|\tfrac{1}{4}\left(a(u+v,u+v)-a(u-v,u-v)\right)\right|^2
\leq\left|\tfrac{1}{2}\left(a(u,u)+a(v,v)\right)\right|^2
\leq\left|\left(a(u,u)a(v,v)\right)\right|
\end{array}\end{equation*}
for $u=P_Lu^{m+\theta}-u^{m+\theta}$ and $v=v_L$, and in the last step we 
used $||\cdot||^2_a\leq||\cdot||^2_{\V}$ (which holds by continuity of the bilinear form cf. Lemma \ref{Th1:WellPosednessSABR}) and the approximation property \eqref{ApproximationPropertyInteger} resp. \eqref{ApproximationPropertyIntegerBetter}.\\
\end{proof}
\subsection{Approximation estimates for CEV: Alternative}\label{Sec:ProofsApproxAlternative}
If we do not assume any additional integrability requirements on our solution up to first order derivatives, we stated in Remark \ref{Prop:ApproximationPropertyIntegerWeights} that we obtain such approximation estimates where the order of approximation depends on the parameter $\beta$. The proof of the estimates
\eqref{ApproximationPropertyIntegerBetter} in Remark \ref{Prop:ApproximationPropertyIntegerWeights} is similar to that of Proposition \ref{Prop:ApproximationPropertyInteger} and is included here:
\begin{proof}[Proof of Remark \ref{Prop:ApproximationPropertyIntegerWeights}]
For $j=1$ and $\mu\in [\max\{-1,-2\beta\},1-2\beta]$ the inequalities in \eqref{eq:CalcProofApproxprop1} become
\begin{align}\label{eq:CalcProofApproxpropWeights1}\begin{split}
&||(u-P_Lu)'||_{L^2(I,x^{\mu/2+\beta})}^2
=\tilde C\sum_{l=L+1}^{\infty}2^{2l}\sum_{k=0}^{2^l}(2^{-l}k)^{\mu+2\beta}|u_k^l|^2\\
&\geq \tilde C 2^{2L(1-\beta)}\sum_{l=0}^{\infty}\sum_{k=0
}^{2^l}2^{2l(\beta-(\mu/2+\beta))}(k)^{\mu+2\beta}|u_k^l|^2\\
&\geq \tilde C 2^{2L(1-\beta)}\sum_{l=0}^{\infty}\sum_{k=0
}^{2^l}(2^{-l}k)^{\mu}|u_k^l|^2
= \tilde C 2^{2L(1-\beta)}||u-P_Lu||_{L^2(I,x^{\mu/2})}^2. 
\end{split}
\end{align}
Similarly, \eqref{eq:CalcProofApproxpropWeights1} after replacing $(u-P_Lu)'$ by $(u-P_Lu)''$ and $(u-P_Lu)$ by $(u-P_Lu)'$ reads
\begin{align}\label{eq:CalcProofApproxpropWeights2}\begin{split}
&||(u-P_Lu)''||_{L^2(I,x^{\mu/2+2\beta})}^2
=\tilde C\sum_{l=L+1}^{\infty}2^{4l}\sum_{k=0}^{2^l}(2^{-l}k)^{\mu+4\beta}|u_k^l|^2\\
&\geq \tilde C 2^{2L(1-\beta)}\sum_{l=0}^{\infty}2^{2l}\sum_{k=0
}^{2^l}2^{2l(\beta-(\mu/2+2\beta))}(k)^{\mu+2\beta}|u_k^l|^2\\
&\geq \tilde C 2^{2L(1-\beta)} \sum_{l=L+1}^{\infty}2^{2l}\sum_{k=0}^{2^l}(2^{-l}k)^{\mu+2\beta}|u_k^l|^2
= \tilde C 2^{2L(1-\beta)}||(u-P_Lu)'||_{L^2(I,x^{\mu/2+\beta})}^2. 
\end{split}
\end{align}
Finally, combining \eqref{eq:CalcProofApproxpropWeights1} and \eqref{eq:CalcProofApproxpropWeights2} yields the last estimate in \eqref{ApproximationPropertyIntegerBetter}.
\end{proof}

\end{document}